%% file: main.tex
\documentclass[UKenglish]{LMCS}

\def\dOi{9(4:12)2013}
\lmcsheading%
{\dOi}
{1--26}
{}
{}
{Jan.~28, 2013}
{Nov.~12, 2013}
{}

\ACMCCS{[{\bf Theory of computation}]: Logic;  Computational complexity and cryptography---Complexity classes}
\subjclass{F.1.3 Complexity Measures and Classes, F.4.1 Math. Logic}

\usepackage[utf8]{inputenc}
\RequirePackage[mathscr]{eucal}
\RequirePackage{amssymb}
\RequirePackage{soul}
\RequirePackage{graphicx}
\RequirePackage{array}
\RequirePackage{listings}
\RequirePackage{subfig}
\RequirePackage{babel}
\RequirePackage{amsmath}

\usepackage{microtype}
\graphicspath{{./graphics/}}

\theoremstyle{plain}

\theoremstyle{definition}

\theoremstyle{remark}

\theoremstyle{plain}

\input ./setup/document-setup.tex

\input ./setup/macros.tex

\usepackage{thmtools} 	

\begin{document}


\title[Definability of linear equation systems over groups and
rings]{Definability of linear equation systems\\over groups and
rings}

\author[A.~Dawar]{Anuj Dawar\rsuper a}	
\address{{\lsuper{a,c}}University of Cambridge, Computer Laboratory}	
\email{\{anuj.dawar, bjarki.holm\}@cl.cam.ac.uk}  

\author[E.~Gr\"adel]{Erich Gr\"adel\rsuper b}	
\address{{\lsuper{b,e}}RWTH Aachen
University, Mathematical Foundations of Computer Science}	
\email{\{graedel, pakusa\}@logic.rwth-aachen.de}  

\author[B.~Holm]{Bjarki Holm\rsuper c}	
\address{\vspace{-18 pt}}	

\author[E.~Kopczynski]{Eryk Kopczynski\rsuper d}	
\address{{\lsuper d}University of Warsaw, Institute of Informatics}	
\email{erykk@mimuw.edu.pl}  

\author[W.~Pakusa]{Wied Pakusa\rsuper e}	
\address{\vspace{-18 pt}}	



\keywords{finite model theory, logics with algebraic operators}

\thanks{{\lsuper{a,c,d,e}}The first and third
authors were supported by EPSRC grant
EP/H026835/1 and the fourth and fifth authors were supported by ESF
Research Networking Programme GAMES. The fourth author was also partially
supported by the Polish Ministry of Science grant N N206 567840.}


\begin{abstract}
Motivated by the quest for a logic for $\Ptime$ and 
recent insights that the descriptive complexity of problems from
linear algebra is a crucial aspect of this problem,
we study the solvability of linear equation systems over
finite groups and rings from the viewpoint of
logical (inter-)definability. All problems that we
consider are decidable in polynomial time, but not
expressible in fixed-point logic with counting. They also
provide natural candidates for a separation of 
polynomial time from rank logics, which extend
fixed-point logics by operators for determining the rank of definable
matrices and which are sufficient for solvability problems
over fields. 

Based on the structure theory of finite rings, we establish logical 
reductions among various solvability problems.
Our results indicate that \emph{all} solvability problems for linear 
equation systems 
that separate fixed-point logic with counting 
from $\Ptime$ can be reduced to solvability over 
commutative rings. 
Moreover, we prove closure properties for classes of queries that
reduce to solvability over rings, 
which provides normal forms
for logics extended with solvability operators. 

We conclude by studying the extent to which fixed-point logic with 
counting can express problems in linear algebra over finite 
commutative rings, generalising known results from~\cite{dawar09logics,holm10thesis,blass02polynomial} 
on the logical definability 
of linear-algebraic problems over finite fields.

\end{abstract}

\maketitle

%
%

%
%
\section*{Introduction}

The quest for a logic for $\Ptime$ \cite{FMTbook,grohe08quest}
is one of the central open problems in both finite model theory and
database theory.
Specifically, it asks whether there is a logic in which a class of
finite structures is expressible if,  and only if, membership in the class
is decidable in deterministic polynomial time.


Much of the research in this area has focused on 
the logic $\fpc$, the extension of inflationary 
fixed-point logic by counting terms. 
In fact, $\fpc$ has been shown to capture 
$\Ptime$ on many natural 
classes of structures, including planar graphs and structures of bounded 
tree-width 
\cite{grohe98fixedpoint,grohe08quest,grohe99treewidth}. 
Recently, it was shown by Grohe~\cite{grohe10minors} that $\fpc$ captures
polynomial time on all classes of graphs with excluded minors, a result that 
generalises most of the previous 
capturing results. 
More recently, it has been shown that $\fpc$ can express important
algorithmic techniques, such as the ellipsoid method for solving
linear programs~\cite{anderson13maximum}.

On the other side, already in 1992, Cai, F\"urer and 
Immerman~\cite{cai92optimal} constructed a graph query 
that can be decided in $\Ptime$, 
but which is not definable 
in $\fpc$. But while this
CFI query, as it is now called,
is very elegant and has led to new insights in many different areas, it can 
hardly be called a natural problem in polynomial time.
Therefore, it was often 
remarked that possibly all \emph{natural} polynomial-time properties of finite 
structures could be expressed in $\fpc$. However, this hope
was eventually refuted in a strong sense by Atserias, 
Bulatov and Dawar~\cite{atserias09affine} who proved that the 
important problem of \emph{solvability of linear equation systems} 
(over any 
finite Abelian group) is not definable in $\fpc$
and that, indeed, the CFI query reduces to this problem. 
This motivates the 
study of the relationship between 
finite model theory and linear algebra, and suggests
that operators from linear algebra could be a  
source of new extensions to fixed-point logic, in an attempt 
to find a logical characterisation of $\Ptime$. 
In~\cite{dawar09logics}, Dawar \etal pursued this direction of study by 
adding operators for expressing the rank of definable matrices over 
finite fields to first-order logic and fixed-point logic. 
They showed that fixed-point logic 
with rank operators ($\FPrk$) can define not only the solvability 
of linear equation systems over 
finite fields, but also 
the CFI query and essentially all other properties that were known to 
separate $\fpc$ from $\Ptime$. 
However, although $\FPrk$ is strictly more expressive than $\fpc$, 
it seems rather unlikely that $\FPrk$ suffices to capture $\Ptime$ 
on the class of all finite structures. 

A natural class of problems that might witness such a separation
arises from linear equation systems over finite domains other than 
fields. Indeed, the results of Atserias, 
Bulatov and Dawar~\cite{atserias09affine} imply that $\fpc$ 
fails to express the solvability of linear equation systems over any 
finite ring. On the other side, it is known that linear equation
systems over finite rings can be solved in polynomial time 
\cite{arvind10classifying}, but it is unclear whether any notion
of matrix rank is helpful for this purpose. We remark in this context
that there are several non-equivalent notions of matrix rank over rings,
but both the computability in polynomial time and the relationship
to linear equation systems remains unclear. Thus, rather than
matrix rank, the solvability of linear equation systems could
be used directly as a source of operators (in the form of 
generalised quantifiers) for extending fixed-point logics.


Instead of introducing a host of new logics, with operators for
various solvability problems, we set out here to investigate whether
these problems are inter-definable. In other words, are they
reducible to each other within $\fpc$? Clearly, if they are,
then any logic that generalises $\fpc$ and can define one,
can also define the others.
We thus study relations between solvability 
problems over (finite) rings, fields and Abelian groups in the context of 
logical many-to-one and Turing reductions, 
i.e., interpretations and generalised quantifiers. In this way, we show that 
solvability both over Abelian groups and over arbitrary (possibly
non-commutative) 
rings reduces to solvability over commutative rings. 
These results indicate that \emph{all} solvability problems for 
linear equation systems that separate $\fpc$ from $\Ptime$ can be reduced
to solvability over commutative rings. 
We also show that 
solvability over commutative rings reduces to solvability over local rings, 
which are the basic building blocks of finite commutative rings. Finally, 
in the other direction, we show that solvability over rings with a
linear order and solvability over 
local rings for which the maximal ideal is generated by $k$ elements,
reduces to
solvability over cyclic groups.  
Further, we prove closure properties for classes of queries that
reduce to solvability over rings, and establish normal forms
for first-order logic extended with operators for 
solvability over finite fields.

While it is known that solvability of linear equation systems over finite 
domains is not expressible in fixed-point logic with counting, it has also been 
observed that  the logic can define many other natural problems from linear 
algebra. For instance, it is known that over finite fields, the inverse to a 
non-singular matrix and the characteristic polynomial of a square matrix can be 
defined in \FPC~\cite{blass02polynomial, dawar09logics}. We conclude this paper 
by studying the extent to which these results can be generalised to finite 
commutative rings. 
Specifically, we use the structure theory of finite commutative rings to show 
that common basic problems in linear algebra over rings reduce to the
respective problems over \emph{local} rings. Furthermore, we show that over 
rings that split into a direct sum of $k$-generated local rings, matrix inverse 
can be defined in $\FPC$. Finally, we show that over the class of Galois rings, 
which are finite rings that generalise finite fields and rings of the form 
$\Zm{p^n}$, there is a formula of $\FPC$ which can define the coefficients of 
the characteristic polynomial of any square matrix. In particular, this shows 
that the matrix determinant is definable in $\FPC$ over such rings.


%
%
\section{Background on logic and algebra}
\label{sec_background}

Throughout this paper, all structures (and in particular, all algebraic
structures such as groups, rings and fields) are assumed to be finite.
Furthermore, it is assumed that all groups are Abelian, unless otherwise
noted.

\subsection{Logic and structures}

The logics we consider in this paper include \emph{first-order logic}
($\FO$) and \emph{inflationary fixed-point logic} ($\FP$) as well as their
extensions by counting terms, which we denote by $\FOC$ and $\FPC$,
respectively. 
We also consider the extension of first-order logic with
operators for deterministic transitive closure, which we denote by
$\FOdtc$. For details see~\cite{ebbinghaus99finite,FMTbook}.

A \emph{vocabulary} $\tau$ is a 
 sequence of relation and constant 
symbols $(R_1, \dots, R_k, c_1, \dots, c_\ell)$ in which every 
$R_i$ has an \emph{arity} $r_i \geq 1$. A
$\tau$-\emph{structure} $\struct A = ( \univ A,
R_1^{\struct A}, \dots, R_k^{\struct A}, c_1^{\struct A}, \dots,
c_\ell^{\struct A})$ consists of a non-empty set $\univ A$, called the
\emph{domain} of $\struct A$, together with relations $R_i^{\struct A}
\subseteq \univ A^{r_i}$ and constants $c_j^{\struct A} \in \univ A$
for each $i \leq k$ and $j \leq \ell$. 
Given a logic~$\logic L$ and a vocabulary $\tau$, we write  $\logic
L[\tau]$ to denote the set of $\tau$-formulas of~$\logic L$. 
A $\tau$-formula $\phi(\tup x)$ with $\card{\tup x} = k$
defines a \emph{$k$-ary query} that takes any $\tau$-structure $\struct A$
to the set $\phi(\tup x)^\struct{A} \defeq \{ \tup a \in \univ{A}^k \sep
\struct A \models \phi[\tup a] \}$. 
To evaluate formulas of counting logics like $\FOC$ and $\FPC$ we associate to 
each $\tau$-structure~$\struct A$ the two-sorted extension $\struct A^+$ of 
$\struct 
A$ by adding as a second sort the standard model 
of arithmetic  $\struct N = (\N, +, \cdot)$.
We assume that in such logics all variables (including the fixed-point 
variables) 
are typed 
and we require that quantification over the second sort is bounded by 
numerical terms in order to guarantee a polynomially bounded range of all 
quantifiers. To relate 
the original structure with the 
second sort we consider counting terms of the form $\# {x} \qsep \phi(x)$ which 
take as value the number of different elements $a \in \univ A$ such that 
$\struct A^+ \models \phi(a)$. For details see~\cite{FMTbook,dawar09logics}.

\makebreak

\makebreak

\noindent\textbf{Interpretations and logical reductions.} 
Consider signatures $\sigma$ and $\tau$ and a logic $\logic L$.  An
\emph{$m$-ary $\logic L$-interpretation of $\tau$ in $\sigma$} is a
sequence of formulas of $\logic L$ in vocabulary $\sigma$ consisting of:
(i) a formula $\delta(\tup x)$; (ii) a formula $\varepsilon(\tup x, \tup
y)$; (iii) for each relation symbol $R \in \tau$ of arity $k$, a formula
$\phi_R(\tup x_1, \dots, \tup x_k)$; and (iv) for each constant symbol $c
\in \tau$, a formula $\gamma_c(\tup x)$, where each $\tup x$, $\tup y$ or
$\tup x_i$ is an $m$-tuple of free variables. We call $m$ the \emph{width}
of the interpretation. We say that an interpretation $\mathcal{I}$
associates a $\tau$-structure $\mathcal{I}(\struct A) = \struct B$ to a
$\sigma$-structure $\struct A$ if there is a surjective map $h$ from the
$m$-tuples $\delta(\tup x) = \{ \tup a \in \univ{A}^m \sep \struct A
\models \delta[\tup a] \}$ to $\struct B$ such that:
\begin{itemize}
	\item $h(\tup a_1) = h(\tup a_2)$ if, and only if, $\struct A \models \varepsilon[\tup a_1, \tup a_2]$;
	
	\item $R^\struct{B}(h(\tup a_1), \dots, h(\tup a_k))$ if, and only if, $\struct A \models \phi_R[\tup a_1, \dots, \tup a_k]$; and
	
	\item $h(\tup a) = c^\struct{B}$ if, and only if, $\struct A \models \gamma_c[\tup a]$.
\end{itemize}

\makebreak

\noindent 


\noindent\textbf{Lindstr\"om quantifiers and extensions.} 
Let $\sigma = (R_1, \dots, R_k)$ be a vocabulary where each relation symbol 
$R_i$ has arity $r_i$, and consider a class
$\structClass K$ of $\sigma$-structures that is closed under isomorphism.

With $\structClass K$ and $m \geq 1$ we associate a 
\emph{Lindstr\"om quantifier}
$\lindstrom K^m$ whose \emph{type} is the tuple $(m; r_1, \ldots, 
r_k)$. For a
logic $\logic L$, we define the extension
$\logic{L}(\lindstrom K^m)$ by adding rules for constructing formulas of the
kind $\lindstrom{K} \tup x_\delta \tup x_\varepsilon \tup x_1 \dots \tup x_k 
\qsep (\delta, \varepsilon, \phi_1, \dots, \phi_k)$,
where $\delta, \varepsilon, \phi_1, \dots, \phi_k$ are $\tau$-formulas, $\tup 
x_\delta$ has length $m$, $\tup x_\varepsilon$ has length $2\cdot m$ and 
each $\tup x_i$ has length $m\cdot r_i$. To define the semantics of this new 
quantifier we associate the interpretation 
$\mathcal{I}=(\delta(\tup x_\delta), 
\varepsilon(\tup x_\varepsilon), (\phi_{i}(\tup x_i))_{1 \leq i \leq k})$ of 
signature $\sigma$ in $\tau$ of width $m$ and we let $\struct A \models
\lindstrom{K} \tup x_\delta \tup x_\varepsilon \tup x_1 \dots \tup x_k 
\qsep (\delta, \varepsilon, \phi_1, \dots, \phi_k)$ if
$\mathcal{I}(\struct A)$ is defined and $\mathcal{I}(\struct A) \in \structClass 
K$ as a $\sigma$-structure
(see~\cite{lindstroem66genQuantifiers, otto97bounded}).
Similarly we can consider the extension of~$\logic L$ by a
collection $\mathbf Q$ of Lindstr\"om quantifiers. The logic $\logic
L(\mathbf Q)$ is defined by adding a rule for constructing formulas with
$Q$, for each $Q \in \mathbf Q$, and the semantics is
defined by considering the semantics for each quantifier $Q \in \mathbf
Q$, as above. 
Finally, we write $\uniformLindstrom K
\defeq \{ Q^m_{\structClass K}  \sep m \geq 1 \}$ to denote the 
\emph{vectorised sequence} of Lindstr\"om quantifiers associated with $\mathcal 
K$
(see~\cite{dawar95generalized}).

\makebreak

\begin{defi}[Logical reductions]
Let $\structClass C$ be a class of $\sigma$-structures and $\structClass D$ a 
class of $\tau$-structures closed under isomorphism. 
\smallskip
\begin{itemize}
	\item 
	$\structClass C$ is said to be \emph{$\logic L$-many-to-one reducible} 
to $\structClass D$ ($\structClass C \logicalReduction{\logic L} \structClass 
D$) if there is an $\logic L$-interpretation $\mathcal{I}$ of $\tau$ in $\sigma$ 
such that for every $\sigma$-structure $\struct A$ it holds that $\struct A \in 
\mathcal C$ if, and only if, $\mathcal{I}(\struct A) \in \mathcal D$. 
	
	\item
	$\structClass C$ is said to be \emph{$\logic L$-Turing reducible} to 
$\structClass D$ ($\structClass C \logicalTurReduction{\logic L} \structClass 
D$) if $\structClass C$ is definable in $\logic L(\uniformLindstrom{D})$. 
\defnend
\end{itemize}
\end{defi}

\noindent
Note that as in the case of usual many-to-one and Turing-reductions,
we have that whenever a class $\structClass C$ is $\logic L$-many-to-one
reducible to a class $\structClass D$, $\structClass C$ is also $\logic
L$-Turing reducible to $\structClass D$.

\subsection{Rings and systems of linear equations}
\label{sec_structure-of-finite-rings-background}

We recall some definitions from commutative and linear algebra, assuming 
that the reader has knowledge of basic algebra and group theory (for
further details see Atiyah~\etal~\cite{atiyah1969introduction}). 
For $m \geq 2$, we write $\Zm m$ to denote the ring of integers modulo $m$. 

\makebreak

\noindent\textbf{Commutative rings.} Let $(R, \cdot, +, 1, 0)$ be a
commutative ring. 
An element $x \in R$ is a \emph{unit} if $x y = y x = 1$ for some $y \in
R$ and we denote by $\units R$ the set of all units. Moreover, we say that
$y$
\emph{divides} $x$ (written $y \divides x$) if $x = y z$ for some $z \in
R$. An element $x \in R$ is \emph{nilpotent} if $x^n = 0$ for some $n \in
\N$, and we call the least such $n \in \N$ the \emph{nilpotency} of $x$.
The element $x \in R$ is \emph{idempotent} if $x^2 = x$. Clearly $0, 1 \in
R$ are idempotent elements, and we say that an idempotent $x$ is
\emph{non-trivial} if $x \notin \{ 0, 1 \}$. Two elements $x,y \in R$ are
\emph{orthogonal} if $xy = 0$.

We say that $R$ is a \emph{principal ideal ring} if every ideal of $R$ is
generated by a single element. An ideal $m \subseteq R$ is called
\emph{maximal} if $m \neq R$ and there is no ideal $m^\prime \subsetneq R$
with $m \subsetneq m^\prime$. A commutative ring $R$ is 
\emph{local} if it contains a unique maximal ideal $m$. Rings that are both
local and principal are called \emph{chain
rings}. For example, all prime rings
$\Zm{p^n}$ are chain rings and so too are all finite fields.
More generally, a \emph{$k$-generated local ring} is a local ring for
which the maximal ideal is generated by $k$ elements. See
McDonald~\cite{mcdonald74finite} for further background.

\begin{rem}
When we speak of a ``commutative ring with a linear order'', then in
general the ordering does not respect the ring operations 
(cp.\ the notion of ordered rings from algebra).	
\end{rem}

\makebreak

\noindent\textbf{Systems of linear equations.}   
We consider systems of linear equations over groups and rings whose 
equations and variables are indexed by arbitrary sets, not necessarily
ordered. In the following, if $I$, $J$ and $X$ are finite and non-empty
sets then an \emph{$I \times J$ matrix} over $X$ is a function $A: I \times
J \rightarrow X$. An \emph{$I$-vector} over $X$ is defined similarly as a
function $\fvec b: I \rightarrow X$.

A system of linear equations over a group $G$ is a pair $(A,\fvec b)$ with
$A \colon I \times J \to \{ 0,1 \}$ and $\fvec b \colon I \rightarrow G$.
By viewing $G$ as a $\Z$-module (i.e.\ by defining the natural
multiplication between integers and group elements respecting $1 \cdot g =
g$, $(n+1)\cdot g = n \cdot g + g$, and $(n-1)\cdot g = n\cdot g - g$),
we write $(A,\fvec b)$ as a matrix equation
$A \cdot \fvec x = \fvec b$, where $\fvec x$ is a $J$-vector of variables
that range over $G$. The system $(A, \fvec b)$ is said to be
\emph{solvable} if there exists a
solution vector $\fvec c\colon J \rightarrow G$ such that $A \cdot \fvec c
= \fvec b$, where we define multiplication of unordered matrices and
vectors in the usual way by $(A \cdot \fvec c)(i) = \sum_{j \in J} A(i,j)
\cdot \fvec c(j)$ for all $i \in I$. We represent linear equation systems
over
groups as finite structures over the vocabulary $\vocLinEqGroup \defeq
( G, A, b , \vocGroup)$, where $\vocGroup \defeq ( +, e
)$ denotes the language of groups, $G$ is a unary relation symbol
(identifying the elements of the group) and $A$, $b$ are two binary
relation
symbols.

Similarly, a system of linear equations over a commutative ring $R$ is a
pair $(A,
\fvec b)$ where $A$ is an $I \times J$ matrix with entries in $R$ and
$\fvec b$ is an $I$-vector over $R$. As before, we usually write $(A, \fvec
b)$ as
a matrix equation $A \cdot \fvec x = \fvec b$ and say that $(A, \fvec b)$
is solvable if there is a solution vector $\fvec c: J \rightarrow R$ such
that $A \cdot \fvec c = \fvec b$. In the case that the ring $R$ is
not commutative, we represent linear systems in the form $A_\ell \cdot
\fvec x + (\fvec x^t \cdot A_r)^t = \fvec b$, where $A_\ell$ is an $I\times 
J$-matrix over $R$ and $A_r$ is a $J \times I$-matrix over $R$, 
respectively.

\makebreak

\noindent
We consider three different ways to
represent linear systems over rings as relational structures. For
simplicity, we just explain the case of linear systems over commutative rings 
here. The encoding of linear systems over non-commutative rings is analogous. 
Firstly, we
consider the case where the ring is part of the
structure. Let $\vocLinEqRing \defeq ( R, A, b, 
\vocRing)$, where $\vocRing = ( +, \cdot, 1, 0 )$ is the
language of rings, $R$~is a unary relation symbol
(identifying the ring elements), and $A$ and $b$ are ternary and binary
relation symbols, respectively. Then a finite $\vocLinEqRing$-structure
$\struct S$ describes the linear equation system $(A^\struct{S}, \fvec
b^\struct{S})$ over the ring $\struct R^{\struct{S} } = (R^\struct{S},
+^\struct{S}, \cdot^\struct{S})$. Secondly, we
consider a similar encoding but with the additional assumption that the
elements of the ring (but not the equations or variables of the equation
systems) are linearly ordered. Such systems can be seen as finite
structures over the vocabulary $\vocLinEqRingOrd \defeq (\vocLinEqRing, 
 \leqslant )$. Finally, we consider linear equation systems
over a fixed ring encoded in the vocabulary: for every ring~$R$, we define
the vocabulary $\vocLinEqRingFixed{R}
\defeq (A_r, b_r \sep r \in R )$, where for each $r \in R$ the symbols
$A_r$ and $b_r$ are binary and unary, respectively. A finite
$\vocLinEqRingFixed{R}$-structure $\struct S$ describes the linear equation
system $(A, \fvec b)$ over $R$ where $A(i,j) = r$ if, and only if, $(i,j)
\in A_r^\struct{S}$ and similarly for~$\fvec b$ (assuming that the
$A_r^\struct{S}$ form a partition of $I\times J$ and that the
$b_r^\struct{S}$ form a partition of $I$). 

Finally, we say that two linear equation systems $\struct S$
and $\struct S'$ are \emph{equivalent}, if
either both systems are solvable or neither system is solvable.

%
%
\section{Solvability problems over different algebraic domains}
\label{sec_les-different-domains}

\newcommand{\giveinappendix}[1]{}
It follows from the work of 
Atserias, Bulatov and Dawar~\cite{atserias09affine} that
fixed-point logic with counting 
cannot express solvability of linear equation systems
(`solvability
problems') over any class of (finite) groups or rings.
In this section we study solvability problems over such different
algebraic domains in terms of logical reductions. Our main result here is 
to show that the solvability problem over groups ($\SolveAbGroup$)
$\FOdtc$-reduces to the corresponding problem over commutative rings
($\SolveRing$) and that the solvability problem over commutative rings
which are equipped with a linear order~($\SolveOrdRing$)
$\FP$-reduces
to the solvability problem over cyclic groups ($\LCON$).
Note that over any non-Abelian group, the solvability problem already is 
\NPtime-complete~\cite{GoRu02}.

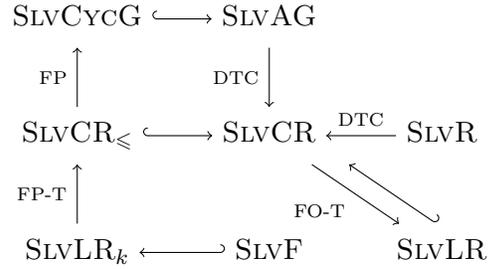
\begin{wrapfigure}{r}{0.50\textwidth}
  	\vspace{-20pt}
	\begin{center} 
		\input{./figures/logical-reductions-solvability}
	\end{center}
  	\vspace{-5pt}
	\caption{Logical reductions between solvability problems.  Curved
arrows ($\hookrightarrow$) denote inclusion of
one class in another.}
	\label{fig_logical-reductions-solvability}
  	\vspace{-4pt}
\end{wrapfigure}
Our methods can be further adapted to show that solvability over
arbitrary (that is, not necessarily commutative) rings
($\SolveGeneralRing$)
$\FOdtc$-reduces
to $\SolveRing\xspace$. We then consider the solvability problem restricted
to special classes of commutative
rings: local rings ($\SolveLocalRing$) and $k$-generated local rings
($\SolvekLocalRing$), which generalises solvability over finite fields
($\SolveField$). The reductions that we establish are
illustrated in 
Figure~\ref{fig_logical-reductions-solvability}. 

In the remainder of this section we describe three of the outlined
reductions: from commutative rings equipped with a linear order to
cyclic groups, from groups to commutative rings, and finally from
general rings to commutative rings. 
To give the remaining reductions from commutative rings to local rings and
from $k$-generated local rings to commutative linearly ordered rings we
need to delve further into the theory of finite commutative rings, which is
the subject of \S\ref{sec_structure-of-finite-rings}. 

\medskip

Let us start by considering the solvability problem over commutative
rings that come with a linear order. We want to construct an
$\FP$-reduction that translates from linear systems over such rings to
equivalent linear equation systems over cyclic groups. Hence, if the ring
is linearly ordered (and in particular if the ring is fixed), this shows
that, up to \FP-definability, it suffices to analyse the solvability
problem over cyclic groups.

\smallskip
The main idea of the reduction, which is given in full detail in the
proof of Theorem~\ref{theorem_reduction-ordings-lcon}, is as follows: for a
ring $R=(R,+,\cdot)$, we consider a decomposition of the additive group
$(R,+)$ into a direct sum of cyclic groups $\langle g_1 \rangle \oplus \cdots
\oplus \langle g_k \rangle$ for appropriate elements $g_i \in R$. Then every
element $r \in R$ can uniquely be identified with a $k$-tuple $(r_1, \dots,
r_k) \in \Zm{\ell_1} \times \cdots \times \Zm{\ell_k}$ where $\ell_i$ denotes
the order of $g_i$ in $(R,+)$, and furthermore, the addition in $R$
translates to component-wise addition (modulo $\ell_i$) in $\Zm{\ell_1} \times
\cdots \times \Zm{\ell_k}$.

Having such a group decomposition at hand, this suggests to treat linear
equations
component-wise, i.e.\ to let variables range over the cyclic summands
$\langle g_i \rangle$ and to split each equation into a set of equations 
accordingly. In
general, however, in contrast to the ring addition, the ring multiplication
will not be compatible with such a decomposition of the group $(R,+)$.
Moreover, an expression of the ring elements with respect to a
decomposition of $(R,+)$ has to be definable in fixed-point logic. To
guarantee this last point, we make use of the given linear ordering.

\begin{thm}
\label{theorem_reduction-ordings-lcon}
$\SolveOrdRing \fpleq \LCON$. 
\end{thm}

\begin{proof}
Consider a system of linear equations $(A, \fvec b)$ over a
commutative ring $R$ of
characteristic $m$ and let $\leqslant$ be a linear order on $R$. 
In the following we
describe a mapping that translates the system $(A, \fvec b)$ into a system 
of equations $(A^\star, \fvec
b^\star)$ over the cyclic group $\Zm{m}$ which is solvable if, and only 
if, $(A, \fvec b)$ has a solution over $R$. Observe that the group $\Zm{m}$
can easily be interpreted in the ring~$R$ in fixed-point
logic, for instance as the subgroup of $(R,+)$ generated by the
multiplicative identity. Indeed, for the purpose of the following
construction we could also identify $\Zm{m}$ with the cyclic group
generated by any element $r \in R$ which has maximal order in $(R,+)$.

Let $\{ g_1, \dots, g_k \} \subseteq R$ be a (minimal) generating set for 
the additive group $(R,+)$ and let~$\ell_i$ denote the order of
$g_i$ in $(R,+)$. Moreover, let us choose the set of generators such that
that $\ell_1 \divides \ell_2 \divides \cdots \divides \ell_k
\divides m$. 
From now on, we identify the group 
generated by $g_i$ with the group $\Zm{m}/\ell_i\Zm{m}$ and thus have
$(R,+) \isom (\Zm{m})^k/(\ell_1\Zm{m} \times \cdots \times
\ell_k\Zm{m})$. 
In
this way we obtain a unique representation for each element $r \in R$ as 
$r=(r_1, \dots, r_k)$ where~$r_i \in \Zm{m}/\ell_i\Zm{m}$. 
Similarly, we can identify variables~$x$ ranging
over $R$ with tuples $x = (x_1, \dots, x_k)$ where $x_i$ ranges
over~$\Zm{m}/\ell_i\Zm{m}$.

\smallskip
To translate a linear equation over $R$ into an equivalent set of 
equations over~$\Zm{m}$, the crucial step is to consider the
multiplication of a coefficient $r \in R$
with a variable $x$ with respect to the chosen representation, 
i.e.\ the
formal expression  $r \cdot x = (r_1, \dots, r_k) \cdot (x_1, \dots, x_k)$.
We observe that the ring multiplication is
uniquely determined by the products of all pairs of generators $g_i
\cdot g_j$, so we let $g_i \cdot g_j = \sum_{y=1}^k c_{y}^{ij} \cdot g_y$,
where $c_y^{ij} \in \Zm{m}/\ell_y\Zm{m}$ for $1 \leq y \leq k$.

Now, let us reconsider the formal expression $r \cdot x$ from above where
$r_i, x_i \in \Zm{m}/\ell_i\Zm{m}$ for $1 \leq i \leq k$, then we have
\[ (r_1 g_1 + \cdots + r_k g_k) \cdot (x_1 g_1 + \cdots + x_k
g_k)=\sum_{i,j\leq k} r_i x_j \sum_{y=1}^k c_y^{ij} g_y = \sum_{y=1}^k
\Big(\sum_{i,j \leq k} r_i x_j c_y^{ij}\Big) g_y. \]

Here, the coefficient of generator $g_y$ in the last expression is an
element in
$\Zm{m}/\ell_y\Zm{m}$, which in turn means that we have to reduce all
summands $r_ix_jc_y^{ij}$ modulo $\ell_y$. To see that this transformation
is sound, we choose $z \in
\Zm{m}$ arbitrary such that $\ord(g_ig_j) \divides z - r_ix_j$. However,
since it holds that $\ell_y \divides \ord(g_ig_j) \cdot c_y^{ij}$ for all
$1 \leq y \leq k$ we conclude that $z c_y^{ij} = r_ix_j c_y^{ij} \mod
\ell_y$ for all $1 \leq i,j,y \leq k$.
Finally, since $\ell_y \divides m$ for all $1 \leq
y \leq k$ we can uniformly consider all terms as taking values in $\Zm{m}$
first, and then reduce the results modulo $\ell_y$ afterwards.

For notational convenience, let us set
$b^{r,y}_j:=
\sum_{i=1}^k r_i c_y^{ij}$, then we can write $rx = (\sum_{j=1}^k b^{r,1}_j
x_j) g_1 + \cdots + (\sum_{j=1}^k b^{r,k}_j x_j) g_k$. Note that the
remaining multiplications between variables $x_j$ and coefficients
$b_j^{r,y}$ are just multiplications in $\Zm{m}/\ell_y\Zm{m}$.

\smallskip
However, for our translation we face a problem, since we cannot 
express that $x_i$ ranges over $\Zm{m}/\ell_i\Zm{m}$ as a linear equation
over $\Zm{m}$.
To overcome this obstacle, let us first drop the requirement completely,
i.e.\ let us consider the multiplication of $R$ in the form given above
lifted to the
group $(\Zm{m})^k$. Furthermore, let $\pi$ denote the natural group
epimorphism which maps $(\Zm{m})^k$ onto $(R,+)$. We claim that for all $r,
x \in (\Zm{m})^k$ we have $\pi(rx)=\pi(r)\pi(x)$. Together with the fact
that $\pi$ is also a group homomorphism from $(\Zm{m})^k$ to $(R,+)$ this
justifies doing all calculations in $\Zm{m}$ first, and reducing the
result to $(R,+)$ via~$\pi$ afterwards.
To see that $\pi(rx)=\pi(r)\pi(x)$ for all $r, x \in (\Zm{m})^k$ let us
denote by $\pi_y$ the natural group epimorphism from
$\Zm{m}\to\Zm{m}/\ell_y\Zm{m}$. Note that for $r = r_1g_1 + \cdots +
r_kg_k$ we have $\pi(r) = \pi_1(r_1)g_1 + \cdots + \pi_k(r_k)g_k$.
Then we have to show that for all $1 \leq y
\leq k$ we have 
$\pi_y( \sum_{i,j \leq k} r_ix_j c_y^{ij} ) =\pi_y(
\sum_{i,j \leq k} \pi_i(r_i) \pi_j(x_j) c_y^{ij} )$. 
For this it suffices to show that for all $i,j \leq k$ we have
$\pi_y( (r_ix_j-\pi_i(r_i)\pi_j(x_j)) c_y^{ij}) = 0$. Let $i \leq j$ (the
other case is symmetric), then $\ord{(g_ig_j)} \divides \ell_i \divides
\ell_j$ and thus by the definition of $c_y^{ij}$ we conclude that $\ell_y
\divides \ell_i c_y^{ij}$. Since $\ell_i \divides \ell_j$ we know that
$\ell_i \divides (r_ix_j-\pi_i(r_i)\pi_j(x_j))$ which yields the
result. 

We are prepared to give the final reduction. In a first step we
substitute each variable $x$ by a tuple of variables $(x_1, \dots, x_k)$
where $x_i$ takes values in $\Zm{m}$. We then translate all terms $rx$ in
the equations of the original system according to the above explanations
and split each equation $e=c$ into a set of $k$ equations $e_1=c_1,
\dots, e_k=c_k$ according to the decomposition of $(R,+)$. We finally 
have to address the issue that the set of new equations is not
equivalent to the original equation $e=c$; indeed, we really want to
introduce the set of equations $\pi_1(e_1)=\pi_1(c_1), \dots,
\pi_k(e_k)=\pi_k(c_k)$. However, this problem can be
solved easily as for all $1 \leq i \leq k$ the linear equation
$\pi_i(e_i)=\pi_i(c_i)$ over $\Zm{m}/\ell_i\Zm{m}$ is equivalent to the
linear equation $\ell_i e_i = \ell_i c_i$ over $\Zm{m}$.

Hence, altogether we obtain a
system of linear equations $(A^\star, \fvec b^\star)$  over $\Zm{m}$
which is solvable if, and only if, the original system $(A, \fvec b)$ has a
solution over $R$.

\medskip

\noindent
We proceed to explain that the mapping $(A, \fvec b) \mapsto (A^\star,
\fvec b^\star)$ can be expressed in $\FP$.
Here, we crucially rely on the given order on $R$ to fix a set of 
generators. More specifically, as we can
compute a set of generators in time polynomial in $\card{R}$, it  follows
from the Immerman-Vardi
theorem~\cite{immerman86relational,vardi82complexity} that there  is an
$\FP$-formula~$\phi(x)$ such that
$\phi(x)^{R} = \{g_1, \dots, g_k \}$ generates $(R, +)$ and  $g_1 \leqslant
\cdots \leqslant g_k$. Having fixed
a set of generators, it is obvious that the  map $\iota: R \rightarrow
(\Zm{m})^k/(\ell_1\Zm{m} \times \cdots
\times \ell_k \Zm{m})$ taking $r \mapsto (r_1, \dots, r_k)$,  is
 $\FP$-definable.
\noindent
Furthermore, the map $(r,y,j) \mapsto b_j^{r,y}$ can easily be formalised
in~$\FP$, since the coefficients are just obtained by performing
a polynomial-bounded number of ring operations. 
Splitting the original system of equations component-wise into $k$ systems
of linear equations, multiplying them with a coefficient $\ell_i$
and combining them again to a single system over $\Zm{m}$ is trivial. 

\smallskip
Finally, we note that a linear system
over the \emph{ring} $\Zm{m}$ can be reduced to an equivalent system
over the \emph{group} $\Zm{m}$, by rewriting terms $ax$ with
$a \in\Zm{m}$ as $x+x+\dots+x$ ($a\text{-times}$).
\end{proof}

Note that in the proof we crucially rely on the fact that we have given a
linear order on the ring~$R$ to be able to fix a set of generators of
the Abelian group $(R,+)$.


\noindent
So far, we have shown that the solvability problem over linearly ordered
commutative rings can be reduced to the solvability problem
over groups. This
raises the question whether a translation in the other direction is 
also possible;
that is, whether we can reduce the solvability problem over groups to
the solvability problem
over commutative rings. 
Essentially, such a reduction requires a logical interpretation of a
commutative ring in a
group, which is what we describe in the proof of the following theorem.

\begin{thm}
\label{theorem_reduction-groups-rings}
$\SolveAbGroup \dtcleq \SolveRing$. 
\end{thm}

\begin{proof}
Let $(A, \fvec b)$ be a system of linear equations over a group
$(G,+_G,e)$, where $A \in \{0,1\}^{I \times J}$ and $\fvec b \in G^I$. For the
reduction, we first construct a commutative ring $\phi(G)$ from $G$ and then
lift $(A,\fvec b)$ to a system of equations $(A^\star, \fvec{b}^\star)$ which
is solvable over $\phi(G)$ if, and only if, $(A,\fvec b)$ is solvable over $G$.

We consider $G$ as a $\Z$-module in the usual way and write $\cdot_{\Z}$ for 
multiplication of group elements by integers. Let $d$ be the least
common multiple of the order of all group elements. Then we
have $\ord_{G}(g) \divides d$ for all $g\in G$,
where $\ord_{G}(g)$ denotes the order of $g$. This allows us to obtain
from
$\cdot_{\Z}$ a well-defined multiplication of $G$ by elements of $\Zm{d}= \{
\eqclass 0 d, \dots, \eqclass {d-1} d \}$ which commutes with group addition. We
write $+_d$ and $\cdot_d$ for addition and multiplication in $\Zm{d}$, where
$\eqclass 0 d$ and $\eqclass 1 d$ denote the additive and multiplicative
identities, respectively. We now consider the set $G \times \Zm{d}$ as a group,
with component-wise addition defined by $(g_1,m_1) + (g_2,m_2) \defeq (g_1 +_G
g_2, m_1 +_d m_2)$, for all $(g_1,m_1), (g_2,m_2) \in G \times \Zm{d}$, and
identity element $0 = (e, \eqclass 0 d)$. We endow $G \times \Zm{d}$ with a
multiplication $\bullet$ which is defined as $(g_1, m_1) \bullet (g_2, m_2)
\defeq \bigl((g_1 \cdot_{\Z} m_2 +_G g_2 \cdot_{\Z} m_1), (m_1 \cdot_d m_2)
\bigr)$.

\smallskip
It is easily verified that this multiplication is associative, commutative and
distributive over $+$. It follows that $\phi(G) \defeq (G \times \Zm{d}, +,
\bullet, 1, 0)$ is a commutative ring, with identity $1 = (e, \eqclass 1 d)$.
For $g \in G$ and $z \in \Z$ we set $\overline g \defeq (g,\eqclass 0 d) \in
\phi(G)$ and $\overline z \defeq (e, \eqclass z d)\in \phi(G)$. Let $\iota: \Z
\cup G \rightarrow \phi(G)$ be the map defined by $x \mapsto \overline x$.
Extending $\iota$ to relations in the obvious way, we write $A^\star \defeq
\iota(A) \in \iota(\Zm d)^{I \times J}$ and $\fvec{b}^\star \defeq \iota(\fvec
b) \in \iota(G)^{I}$. 

\makebreak
\noindent
\emph{Claim.}
The system $(A^\star, \fvec{b}^\star)$ is solvable over $\phi(G)$ if, and only
if, $(A, \fvec b)$ is solvable over $G$.

\makebreak
\noindent
\emph{Proof of claim.}
In one direction, observe that a solution $\fvec s$ to $(A,\fvec b)$
gives the
solution $\iota(\fvec s)$ to $(A^\star, \fvec{b}^\star)$. For the other
direction, suppose that $\fvec s \in \phi(G)^J$ is a vector such that $A^\star
\cdot \fvec s = \fvec b^\star$. Since each element $(g,\eqclass m d) \in
\phi(G)$ can be written uniquely as $(g,\eqclass m d) = \overline g + \overline
m$, we write $\fvec s = \fvec s_g + \fvec s_n$, where $\fvec s_g \in \iota(G)^J$
and $\fvec s_n \in \iota(\Zm{d})^J$. Observe that we have $\overline g \bullet
\overline m \in \iota(G) \subseteq \phi(G)$ and $\overline n \bullet \overline m
\in \iota(\Zm d) \subseteq \phi(G)$ for all $g \in G$ and $n,m \in \Z$. Hence,
it follows that $A^\star \cdot \fvec s_n \in \iota(\Zm{d})^I$ and $A^\star \cdot
\fvec s_g \in \iota(G)^I$. Now, since $\fvec b^\star \in \iota(G)^I$, we have
$\fvec b^\star = A^\star \cdot \fvec s = A^\star \cdot \fvec s_g + A^\star \cdot
\fvec s_n = A^\star \cdot \fvec s_g$. Hence, $\fvec s_g$ gives a solution to
$(A, \fvec b)$, as required.

\smallskip

All that remains is to show that our reduction can be formalised as an 
interpretation in $\FOdtc$.
Essentially, this comes down to showing that 
the ring $\phi(G)$ can be interpreted in $G$ by formulas of $\FOdtc$. By
elementary group theory, we know that for elements $g \in G$ of maximal
order we have $\ord{(g)} =
d$. It is not hard to see that the set of group elements of maximal order
can
be defined in $\FOdtc$; for example, for any fixed $g \in G$ the set of 
elements of the form $n \cdot g$ for $n \geq 0$ is $\FOdtc$-definable as a
reachability query in the deterministic directed graph $E = \{ (x,y) : y = x + g 
\}$.
Hence, we can interpret $\Zm{d}$ in $G$, and as 
on ordered domains, $\FOdtc$ expresses all $\Logspace$-computable queries 
(see e.g.~\cite{FMTbook}) the multiplication of $\phi(G)$ is also
$\FOdtc$-definable, which completes the proof.
\end{proof}

\noindent
We conclude this section by discussing the solvability problem over
general (i.e.\ not necessarily commutative) rings $R$. Over such rings,
linear equation systems have a representation of the form
$A_\ell \cdot \fvec x + (\fvec x^t \cdot A_r)^t = \fvec b$ where $A_\ell$ and
$A_r$ are two coefficient matrices over $R$. This representation takes into
account the difference between left and right multiplication of variables
with coefficients from the ring.

First of all, if the ring comes with a linear ordering, then it is
easy to adapt the proof of
Theorem~\ref{theorem_reduction-ordings-lcon} for the case of
non-commutative rings. Hence, in this case we obtain again an
$\FP$-reduction to the solvability problem over cyclic groups. Moreover, in
what follows we are going to establish a $\FOdtc$-reduction from the
solvability problem over general rings to the solvability problem over
commutative rings $R$. These results indicate that from the viewpoint of
$\FP$-definability the solvability problem does not become harder when
considered over arbitrary (i.e.\ possibly non-commutative) rings.

As a technical preparation, we first give a
first-order interpretation that transforms a linear equation systems
over $R$ into an equivalent system with the following
property: the linear equation system is solvable if, and only if, the
solution
space contains a \emph{numerical solution}, i.e.\ a solution over $\Z$. 

%


\begin{lem}
\label{lemma_solvability-numerical-solutions}
There is an $\FO$-interpretation $\mathcal{I}$ of $\vocLinEqRing$ in 
$\vocLinEqRing$ such that for every linear equation system $\struct S:
A_\ell \cdot
\fvec x + (\fvec x \cdot A_r)^t = \fvec b$ over $R$,
$\mathcal{I}(\struct{S})$
describes a
linear equation system $\struct{S^\star}: A^\star \cdot_\Z \fvec{x}^\star =
\fvec{b}^\star$ over the $\Z$-module $(R, +)$ such that
$\struct S$ is solvable over $R$ if, and only if, $\struct S^\star$ has a
solution over $\Z$. 
\end{lem}

\begin{proof}[Proof]
Let 
$A_\ell \in R^{I \times J}$, $A_r \in R^{J \times I}$, and $\fvec b
\in R^I$. By duplicating each variable we can assume that for every $j \in
J$ we have that $A_\ell(i,j)=0$ for all $i \in I$, or $A_r(j,i)=0$ for all
$i \in I$, i.e.\ we assume that each variable occurs only with either
left-hand or right-hand coefficients.
For $\struct S^\star$, we introduce for each variable $x_j$
($j \in J$) and each element $s \in R$ a new variable $x_j^s$, i.e.\ the
index set for the variables of $\struct S^\star$ is $J \times R$. 
Finally, we replace all terms of the form $rx_j$ by $\sum_{s \in R}
rsx_j^s$, and similarly, terms of the form $x_jr$ by $\sum_{s \in R}
srx_j^s$. If we let the new variables $x_j^s$ take values in $\Z$,
then we obtain a new linear equation system of the desired form
$\struct{S^\star}: A^\star \cdot_\Z \fvec{x}^\star = \fvec{b}^\star$ over
the $\Z$-module $(R,+)$.
It is easy to see that this transformation can be formalised by an
\FO-interpretation $\mathcal{I}$.

\medskip
Finally we observe that the newly constructed linear equation system
$\struct {S^\star}$ is equivalent to the original system $\struct S$.
To see this, assume that $\fvec x \in R^J$ is a solution of the
original system. By setting $x_j^s = 1$ if $x_j = s$ and by setting $x_j^s
= 0$ otherwise, we obtain a solution $\fvec {x^\star} \in \Z^{R \times J}$
of the system~$\struct {S^\star}$. For the other direction, assume that
$\fvec {x^\star} \in \Z^{R \times J}$ is a solution of $\struct{S^\star}$.
Then we set $x_j := \sum_{s \in R} sx_j^s = \sum_{s \in R} x_j^ss$ to get a
solution $\fvec x$ for the system $\struct{S}$.
\end{proof}

\noindent
By Lemma~\ref{lemma_solvability-numerical-solutions}, 
we can restrict to linear equation systems $(A, \fvec b)$ over the
$\Z$-module $(R,+)$ where variables take values in $\Z$. However, since
$\Z$ is an infinite domain, we let $d=\max\{\ord(r) : r \in R\}$ denote the
maximal order of elements in the group $(R,+)$. Then we can also treat $(A,
\fvec b)$ as an equivalent linear equation system over the $\Zm{d}$-module
$(R,+)$.

\smallskip
 At this point, we reuse
our construction from                                          
Theorem~\ref{theorem_reduction-groups-rings} to obtain a linear
system $(A^\star, \fvec b^\star)$ over the commutative ring $R^\star
\defeq \phi((R,+))$, where $A^\star \defeq \iota(A)$ and $\fvec
b^\star \defeq \iota(\fvec b)$. 
We claim that $(A^\star, \fvec b^\star)$ is solvable over
$R^\star$ if, and only if, $(A, \fvec b)$ is solvable over $R$. For the
non-trivial direction, suppose $\fvec s$ is a solution to $(A^\star, \fvec
b^\star)$ and decompose $\fvec s = \fvec s_g + \fvec s_n$ into group
elements and number elements, as explained in the proof of
Theorem~\ref{theorem_reduction-groups-rings}.
Recalling that $\overline r_1 \bullet \overline r_2 = 0$ for all $r_1,
r_2 \in R$, it follows that $A^\star \bullet (\fvec s_g + \fvec s_n) =
A^\star \bullet \fvec s_n = \fvec b^\star$. Hence, there is a solution
$\fvec s_n$ to $(A^\star, \fvec b^\star)$ that consists \emph{only} of
number elements, as claimed. Thus we obtain:

\begin{thm}
 $\SolveGeneralRing \dtcleq \SolveRing$.
\end{thm}

%
%
\section{The structure of finite commutative rings}
\label{sec_structure-of-finite-rings}

In this section we study 
structural properties of (finite) commutative
rings and present the remaining reductions for solvability outlined in
\S\ref{sec_les-different-domains}: from commutative rings to local rings,
and from $k$-generated local rings to commutative rings with a linear
order. 
Recall that a commutative ring $R$ is local if it
contains a unique maximal ideal $m$. 
The importance of the notion of local rings comes from the fact that
they are the basic building blocks of finite
commutative rings. We start by summarising some of their useful
properties.

\begin{prop}[Properties of (finite) local
rings]
\label{proposition_local-rings-properties} 
Let $R$ be a finite commutative ring.
\begin{itemize}
	\item If $R$ is local, then the unique maximal ideal is $m = R
\setminus \units R$.
	\item $R$ is local if, and only if, all idempotent elements in
$R$ are trivial. 
	\item If $x \in R$ is idempotent then $R = x \cdot R \oplus
(1-x) \cdot R$ as a direct sum of rings.
	\item If $R$ is local then its cardinality (and in particular its
characteristic) is a prime power.
\end{itemize}
\end{prop}
\begin{proof} The first claim follows directly by the uniqueness of
the maximal ideal $m$. For the second part, assume $R$ is local but
contains a non-trivial idempotent element $x$, i.e.\ $x (1-x) = 0$ but
$x\not=0,1$. In this case
$x$ and $(1-x)$ are two non-units distinct from $0$, hence $x, (1-x) \in
m$. But then $x + (1-x) = 1 \in m$ which yields the contradiction. On
the other hand, if $R$ only contains trivial idempotents, then we claim
that every non-unit in $R$ is nilpotent: assume that $x\not=0$ is a
non-unit which is not nilpotent, then $x^{n+km}=x^n$ for some $m,
n\geq 1$ and all $k \geq 1$ because $R$ is finite. In particular,
\[ x^{nm} \cdot x^{nm} = x^n \cdot x^{nm-n} \cdot x^{nm} = x^{n+nm}
\cdot x^{nm-n} = x^n \cdot x^{nm-n} = x^{nm}.\]
Since $x^{nm}\not=1$ we have $x^{nm}=0$ which is a contradiction to
our assumption that $x$ is not nilpotent. Hence we have that $x$ is a
non-unit if, and only if, $x$ is nilpotent. Knowing this, it is easy
to verify that also sums of non-units are non-units, which implies that
the set of non-units forms a unique maximal ideal in $R$.

For the third part, assume $x \in R$ is  idempotent. Then $(1-x)^2 =
(1 - 2x + x^2) = (1-x)$ so $(1-x)$ is also idempotent. Furthermore,
as $x(1-x) = 0$ we see that $x$ and $(1-x)$
are orthogonal, and since $x + (1-x) = 1$, any element $r \in R$
can be expressed as $r = rx + r(1-x)$, so we conclude that $R = x R
\oplus (1-x) R$.


Finally, let $R$ be local and suppose $\card{R} = p^k n $ where $p
\nmid n$. We want to show that $n =1$. Otherwise $I_p = \lbrace r \in
R \sep p^k r = 0 \rbrace$ and $I_n = \lbrace r \in R \sep n r = 0
\rbrace$ would be two proper distinct ideals. To see this, let 
$x,y \in \Z$ with $x p^k + y n = 1$. We first show that $I_p\cap
I_n=\{0\}$. Assume $p^k r = 0 = n r$ for some $r \in R$, then $x p^k r
+ y n r = 0$ and hence $r = 0$.
Furthermore we show that $R = I_p + I_c$: for each $r \in R$ we
have that $n r \in I_p, p^k r \in I_n$ and so $y n r + x p^k r = r \in
(I_p + I_n)$. This shows that $R$ does not contain a unique maximal ideal,
and so $R$ was not local. 
\end{proof}

\noindent
By this proposition we know that
finite commutative rings can be decomposed into local summands that
are principal ideals generated by pairwise orthogonal
idempotent elements. Indeed, this decomposition is unique (for more
details, see e.g.~\cite{bini02finite}).
\noindent

\begin{prop}[Decomposition into local rings]
\label{proposition_decomposition-into-local-rings}
Let $R$ be a (finite) commutative ring. Then there is a unique set 
$\ringBase R \subseteq R$ of pairwise orthogonal idempotent elements
for which it holds that (i) $e \cdot R$ is local for each $e \in
\ringBase R$; (ii) $\sum_{e \in \ringBase R} e = 1$;  and (iii) $R =
\bigoplus_{e \in \ringBase R} e \cdot R$. 
\end{prop}

\noindent
We next show that the ring decomposition $R =
\bigoplus_{e \in \ringBase R} e \cdot R$  is $\FO$-definable. As a first
step, we note that $\ringBase R$ (the
\emph{base} of $R$) is $\FO$-definable over $R$. 

\begin{lem}
\label{lemma_ring-idempotent-base} 
There is a formula $\phi(x) \in \FO(\vocRing)$ such that 
$\phi(x)^R = \ringBase{R}$ for every (finite) commutative ring $R$.
\end{lem}

\begin{proof}
We claim that $\ringBase{R}$ consists precisely of those non-trivial
idempotent elements of $R$ which cannot be expressed as the sum of two
orthogonal non-trivial idempotent elements. To establish this claim,
consider an element $e \in \ringBase{R}$ and suppose that $e = x + y$
where $x$ and $y$ are orthogonal non-trivial idempotents. It follows
that $e$ is different from both $x$ and $y$, since if $e = x$, say,
then $y = e - x = 0$ and similarly when $e = y$. Now $ex = xe = x(x+y)
= x^2 + xy = x$ and, similarly, $ey = y$. Since both $ex$ and $ey$ are
idempotent elements in $eR$, it follows that $ex, ey \in \lbrace 0, e
\rbrace$, since $eR$ is local with identity $e$ and contains no
non-trivial idempotents. But by the above we know that $ex = x \neq e$
and $ey = y \neq e$, so $ex = ey = x = y = 0$. This contradicts the
fact that $e = x + y$ is non-trivial, so the original assumption must
be false.

Conversely, suppose $x \in R$ is a non-trivial idempotent element that
cannot be written as the sum of two orthogonal non-trivial
idempotents. Writing $\ringBase{R} = \lbrace e_1, \dots, e_m \rbrace$,
we get that  
\[
	x = x (1) = x (e_1 + \cdots + e_m) = xe_1 + \cdots + xe_m.
\]

\noindent
Each $xe_i$ is an idempotent element of $e_iR$ and since $e_iR$ is
local, $xe_i$ must be trivial. Hence, there are distinct $f_1, \dots,
f_n \in \ringBase{R}$, with $n \leq m$, such that $x = f_1 + \cdots +
f_n$. But since $x$ cannot be written as a sum of two (or more)
non-trivial idempotents, it follows that $n = 1$ and $x \in
\ringBase{R}$, as claimed.

\makebreak

\noindent
Now it is straightforward to write down a first-order formula that
identifies exactly all non-trivial idempotent elements that are not
expressible as the sum of two non-trivial orthogonal idempotents.
Moreover, if $R$ was already local then trivially $\ringBase{R} = \{1\}$.
To test for locality, it suffices
by Proposition~\ref{proposition_local-rings-properties} to check
whether all idempotent elements in $R$ are trivial and this can be
expressed easily in first-order logic.
\end{proof}

\noindent
The next step is to show that the canonical mapping $R \to \bigoplus_{e \in
\ringBase R} e \cdot R$ can be defined in $\FO$. To this end, recall from
Proposition~\ref{proposition_local-rings-properties} that
for every $e \in \ringBase R$ (indeed, for any idempotent element $e
\in R$), we can decompose the ring $R$ as $R = e
\cdot R \oplus (1-e) \cdot R$. This fact allows us to define for all base
elements $e \in \ringBase R$ the projection of elements $r \in R$ onto the
summand $e\cdot R$ in first-order logic, without having to keep track of
all local summands simultaneously.

\begin{lem}
\label{lemma_projection-local-ring}
There is a formula $\psi(x, y, z) \in \FO(\vocRing)$
such that for all rings $R$, $e \in \mathcal{B}(R)$ and $r, s \in R$, it
holds that $(R, e, r, s) \models \psi$ if, and only if, $s$ is the
projection of $r$ onto $e \cdot R$. 
\end{lem}

\begin{proof}
The formula can simply be given as
\[
	\psi(x, y, z) := \exists a, b \,\big( (y =
x\cdot a + (1-x) \cdot b) \wedge (z = x \cdot a)\big).   
\]
\end{proof}

\noindent
It follows that any relation over $R$ can be decomposed in first-order 
logic according to the decomposition of $R$ into local summands. In
particular, a linear equation system $(A \sep \fvec b)$ over $R$ is
solvable if, and only if, each of the projected linear equation systems
$(A^e \sep \fvec b^e)$ is solvable over $eR$. Hence, we obtain:

\begin{thm}
\label{reduction_solvability-rings-to-local-rings}
$\SolveRing \foleqTur \SolveLocalRing$. 
\end{thm}

\noindent
We now want to exploit the algebraic structure of finite local rings
further. In \S\ref{sec_les-different-domains} we proved that solvability
over rings with a linear ordering can be reduced in fixed-point logic to
solvability over cyclic groups. This naturally raises the question:
which classes of rings can be linearly ordered in fixed-point logic?
By Lemma~\ref{lemma_projection-local-ring}, we know that for this
question it suffices to focus on local rings, which have a
well-studied structure. The most basic local rings are the
rings $\Zm{p^n}$, and the natural ordering of such rings can be easily
defined in~$\FP$ (since the additive group of $\Zm{p^n}$ is cyclic).
Moreover, the same holds for finite fields as they have a cyclic
multiplicative group~\cite{holm10thesis}. 

In the following lemma, 
we are able to generalise these
insights in a strong sense: for any fixed $k\geq 1$ we
can define an ordering on the 
class of all local rings for which the maximal ideal is
generated by at most $k$ elements. We refer to such rings as \emph{$k$-generated local
rings}. Note that for $k=1$ we obtain the notion
of chain rings which include all finite fields and rings of the form 
$\Zm{p^n}$. For increasing values of $k$ the
structure of $k$-generated local rings 
becomes more and more
sophisticated. For instance, the ring $R_k =
\Zm{2}[X_1,\dots,X_k]/(X_1^2,\dots,X_k^2)$ is a $k$-generated
local ring which is not $(k-1)$-generated.

\begin{lem}[Ordering $k$-generated local rings]
\label{lemma_ordering-local-pir} 
There is an $\FP$-formula $\phi(x,z_1,\dots,z_k; v,w)$
such that for all $k$-generated local rings $R$ there are $\alpha,
\pi_1,\dots,\pi_k \in R$ such that
\[
	\phi^R(\alpha/x, \tup \pi/\tup z; v,w) = \{ (a,b) \in R \times
R \sep (R, \alpha,\tup \pi; a,b) \models \phi \},\]
is a linear order on $R$.
\end{lem}

\begin{proof} 
First of all, there are $\FP$-formulas $\phi_u(x), \phi_m(x),
\phi_g(x_1,\dots,x_k)$ that define in any $k$-generated local ring $R$ the
set of units, the maximal ideal $m$ (i.e.\ the set of non-units) and the
property of being a set of size $\leq k$ that generates the ideal $m$,
respectively. More precisely, for all $(\pi_1,\dots,\pi_k)
\in \phi_g^R$
we have that $\sum_i\pi_i R = \phi_m^R=m$ is the maximal ideal of $R$
and
$\units{R} = \phi_u^ R = R \setminus m$. In particular there is a
first-order interpretation of the field $F \defeq R/m$ in $R$.

\medskip
The idea of the following proof is to represent the elements of $R$ as
polynomial
expressions of a certain kind. Let $q \defeq
\card{F}$ and define $\Gamma(R) \defeq \{ r \in R : r^q = r\}$. It
can be seen that $\Gamma(R)\setminus \{0\}$ forms a multiplicative group
which is known as the \textit{Teichm\"{u}ller coordinate
set}~\cite{bini02finite}. Now, the map $\iota: \Gamma(R) \rightarrow
F$ defined by $r \mapsto r+m$ is a bijection. Indeed,
for two different units $r,s \in \Gamma(R)$ we have $r-s \notin
m$. Otherwise, we would have $r-s = x$ for some $x \in m$ and thus $r
= (s+x)^q = s + \sum_{i=1}^q {q \choose i} x^is^{q-i} $. Since $q \in
m$ and $r-s=x$ we obtain that $x = xy$ for some $y \in m$. Hence
$x(1-y) = 0$ and since $(1-y) \in \units R$, as in a local ring the sum
of a unit and a non-unit is always a unit, this means $x=0$.

As explained above, we can define in $\FP$ an order
on $F$ by fixing a generator $\alpha \in \units F$ of the cyclic
group $\units F$. Combining this order with $\inv{\iota}$, we obtain
an $\FP$-definable order on~$\Gamma(R)$. 
The importance of $\Gamma(R)$ now lies in the fact that every ring element
can
be expressed as a polynomial expression over a set of $k$ generators of
the maximal ideal $m$ with coefficients lying in $\Gamma(R)$. To be
precise, let $\pi_1, \dots, \pi_k \in m$ be a set of generators for $m$,
i.e.\ $m = \pi_1R + \cdots + \pi_kR$, where each $\pi_i$
has nilpotency $n_i$ for $1 \leq i \leq k$.
We claim that we can express $r \in R$ as
\begin{align*}
r = \sum_{\substack{(i_1,\dots,i_k) \lleq (n_1,\dots,n_k)}}
a_{i_1\cdots i_k}
\pi_1^{i_1}\cdots\pi_k^{i_k},\quad \text{with } a_{i_1\cdots i_k}
\in \Gamma(R). \tag{P}\label{equ:frs:polyexpr}
\end{align*}

\noindent To see this, consider the following recursive algorithm: 
\smallskip
\begin{itemize}
	\item If $r \in \units{R}$, then for a unique $a \in
\Gamma(R)$ we have that $r \in a+m$, so $r = a + (\pi_1  r_1 + \cdots +
\pi_kr_k)$ for some $r_1,\dots,r_k \in R$ and we continue with
$r_1,\dots,r_k$. 
	\item Else $r \in m$, and $r = \pi_1  r_1 + \cdots +
\pi_kr_k$ for some $r_1,\dots,r_k \in R$; continue
with $r_1,\dots,r_k$. 
\end{itemize}
\smallskip
\noindent
Observe that for all pairs $a, b\in \Gamma(R)$ there exist elements
$c \in \Gamma(R), r \in m$ such that
$a\pi_1^{i_1}\cdots\pi_k^{i_k} + b\pi_1^{i_1}\cdots\pi_k^{i_k} =
c\pi_1^{i_1}\cdots\pi_k^{i_k} + r\pi_1^{i_1}\cdots\pi_k^{i_k}$. Since
$\pi_1^{i_1} \cdots \pi_k^{i_k} = 0$ if $i_\ell \geq n_\ell$ for some
$1 \leq \ell \leq k$, the process is guaranteed to stop and the claim
follows.

\medskip
Note that this procedure neither yields a polynomial-time algorithm
nor do we obtain a \emph{unique} expression, as for
instance, the choice of elements $r_1, \dots, r_k \in R$ (in both
recursion steps) need not to be unique. However, knowing only the
existence of an expression of this kind, we can proceed as
follows. For any sequence of exponents $(\ell_1, \dots,
\ell_k) \lleq (n_1,\dots,n_k)$ define the ideal $R[\ell_1, \dots,
\ell_k] \trianglelefteq R$ as the set of all elements having an
expression of the form~(\ref{equ:frs:polyexpr}) where $a_{i_1\cdots i_k} =
0$ for all $(i_1,\dots,i_k) \lleq (\ell_1, \dots, \ell_k)$.

It is clear that we can define the ideal $R[\ell_1, \dots, \ell_k]$ in
$\FP$. Having this, we can use the following recursive procedure to
define a unique expression of the form~(\ref{equ:frs:polyexpr}) for
all $r \in R$:
\begin{itemize}
 \item Choose the minimal $(i_1,\dots,i_k) \lleq (n_1,\dots,n_k)$
such that $r= a\pi_1^{i_1}\cdots\pi_k^{i_k} + s$ for some (minimal) $a
\in \Gamma(R)$ and $s \in R[i_1, \dots, i_k]$. Continue the process
with $s$.
\end{itemize}

\noindent
Finally, the lexicographical ordering induced by
the ordering on $n_1 \times \cdots \times n_k$ and the
ordering on $\Gamma(R)$ yields an $\FP$-definable order on $R$ (where we
use the parameters to fix a generator of $\units F$ and a set of
generators of $m$).
\end{proof}

\begin{cor}
\label{reduction_solvability-chain-rings-to-ordered-rings}
$\SolvekLocalRing \fpleqTur \SolveOrdRing \fpleq \LCON$. 
\end{cor}

%
%
\section{Solvability problems under logical reductions}
\label{sec_reductions-to-solvability}

\newcommand{\RelativeSolve}[3]{\Sigma_{#1}^{#3}(#2)\xspace}
\newcommand{\rtsqr}[1]{\RelativeSolve{\FO}{#1}{\text{qf}}}
\newcommand{\rtsfo}[1]{\RelativeSolve{\FO}{#1}{}}
\newcommand{\rtstfo}[1]{\RelativeSolve{\FO}{#1}{\text{T}}}
\newcommand{\clsles}[1]{\Solve(#1)}

In the previous two sections we studied reductions between solvability 
problems over different algebraic domains. Here we change our perspective
and investigate \emph{classes of queries} that are reducible to the
solvability problem over a fixed commutative ring. Our motivation for this
work was to study extensions of first-order logic with generalised
quantifiers which express solvability problems over rings. In particular,
the aim was to establish various \emph{normal forms} for such logics.
However, rather than defining a host of new logics in full detail, we state
our results in this section in terms of closure properties of classes of
structures that are themselves defined by reductions to solvability
problems. We explain the connection between the specific closure 
properties and the corresponding logical normal forms below.

To state our main results formally, let $R$ be a \emph{fixed} commutative
ring. We write
$\clsles{R}$ to denote the solvability problem over~$R$, as a class of
$\vocLinEqRingFixed{R}$-structures. Let $\rtsqr{R}$ and $\rtsfo{R}$ denote
the classes of queries that are reducible to $\clsles{R}$ under
quantifier-free and first-order many-to-one reductions, respectively. Then
we show that $\rtsqr{R}$ and $\rtsfo{R}$ are closed under first-order
operations for any commutative ring $R$ of \emph{prime-power
characteristic},  i.e.\ $\text{char}(R)=p^k$ for a prime $p$ and an integer
$k \geq 1$. In particular, we have that $\rtsqr{R}$ contains \emph{any}
$\FO$-definable query in such a case. Furthermore, we prove that if
$R$ has \emph{prime characteristic}, i.e.\ $\text{char}(R)=p$ for a prime
$p$, then $\rtsqr{R}$ and $\rtsfo{R}$ are closed under oracle queries. 
Thus, if we denote by $\rtstfo{R}$ the class
of queries reducible to $\clsles{R}$ by first-order Turing reductions, 
then for all commutative rings $R$ of prime characteristic the three
solvability reduction classes coincide, i.e.\ we have $\rtsqr{R} =
\rtsfo{R} =\rtstfo{R}$. 

To relate these results to logical normal forms,  we let $\structClass D =
\clsles{R}$ and write $\foslv R \defeq \fo(\uniformLindstrom{D})$ to denote
first-order logic extended by Lindstr\"om quantifiers
expressing solvability over $R$. Then the closure of $\rtsfo{R}$ under
first-order operations amounts to showing that the fragment of $\foslv R$
which consists of formulas without \emph{nested} solvability
quantifiers has a normal form which consists of a single application of 
a solvability quantifier to a first-order formula. Moreover, for the case
when $R$ has prime characteristic, the closure of $\rtsqr{R}= \rtsfo{R}$
under first-order oracle queries amounts to showing that nesting of
solvability quantifiers can be reduced to a single quantifier. It follows
that $\foslv R$ has a strong normal form: one application of a
solvability quantifier to a \emph{quantifier-free} formula
suffices. It remains as an interesting open question whether the closure
properties
we establish here can be extended to the case of general commutative rings,
i.e.\ to rings $R$ whose characteristic is divisible by two different
primes, e.g.\ to $R=\Zm{6}$.

Throughout this section, the reader should keep in mind that for logical
reductions to the solvability problem over a \emph{fixed} ring $R$, we can
safely drop all formulas $\varepsilon(\bx,\by)$ in interpretations
$\mathcal{I}$ which define the equality-congruence on the domain of the
interpreted structure: indeed, duplicating equations or
variables does not affect the solvability of a given linear equation
system.

\subsection{Closure under first-order operations}

Let $R$ be a fixed commutative ring of characteristic $m$. In
this section we
prove the closure of $\rtsqr{R}$ and $\rtsfo{R}$ under first-order
operations for the case that $m$ is a prime power.
To this end, we need to establish a couple of technical
results. Of particular importance is the following key lemma, which
gives a simple normal form for linear equation systems: up
to quantifier-free reductions, we can restrict ourselves to
linear systems over rings $\Zm m$,
where the constant term $b_i$ of every linear equation $(A\cdot \fvec
x)(i)=b_i$ is $b_i=1 \in \Zm m$, and the same holds for all the
coefficients, i.e.\ $A(i,j)=1$, for all $i,j \in I$.
The proof of the lemma crucially relies on the fact that
the ring $R$ is fixed. Recall that $m$ denotes the characteristic of the ring 
$R$.

\begin{lem}[Normal form for linear equation
systems]
\label{lemmanormalformles}
There is a quantifier-free interpretation $\mathcal{I}$ of 
$\vocLinEqRingFixed{\Zm m}$ in $\vocLinEqRingFixed{R}$ such that for all
$\vocLinEqRingFixed{R}$-structures $\struct{S}$ it holds that
\begin{itemize}
	\item $\mathcal{I}(\struct{S})$ is an equation system 
$(A, \fvec b)$ over $\Zm m$, where $A$ is a $\{0,1\}$-matrix and 
$\fvec b = \fvec{1}$; and 
	\item  $\struct{S} \in \clsles{R}$ if, and only 
      if, $\mathcal{I}(\struct{S}) \in \clsles{\Zm m}$. 
\end{itemize}
\end{lem}

\begin{proof}
We describe $\mathcal I$ as the composition of
three quantifier-free transformations: the first
one maps a system $(A, \fvec b)$ over $R$ to an equivalent
system $(B, \fvec c)$ over $\Zm m$, where $m$ is the
characteristic of $R$. Secondly, $(B, \fvec c)$ is mapped to an
equivalent system $(C, \fvec 1)$ over $\Zm m$.
Finally, we
transform $(C, \fvec 1)$ into an equivalent system $(D, \fvec 1)$ over
$\Zm m$, where $D$ is a $\{0,1\}$-matrix.  The first transformation is
obtained by adapting the proof of
Theorem~\ref{theorem_reduction-ordings-lcon}. It can be seen that
first-order quantifiers and fixed-point operators are not needed if
$R$ is fixed. 

For the second transformation,
suppose that $B$ is an $I \times J$ matrix and $\fvec c$ a vector
indexed by $I$. We define a new linear equation system $\struct T$
which has in addition to all the
variables that occur in $\struct S$, a new variable $v_e$ for every $e
\in I$ and a new variable $w_r$ for every $r \in R$. For every element
$r \in \Zm m$, we include in $\struct T$ the equation $(1-r)w_1 + w_r=
1$. It can be seen that this subsysem of equations has a unique
solution given by $w_r = r$ for all $r \in \Zm m$. Finally, for every
equation $\sum_{j \in J} B(e,j) \cdot x_j = \fvec c(e)$ in $\struct S$
(indexed by $e \in I$) we include in $\struct T$ the two equations
$v_e + \sum_{j \in J} B(e,j) \cdot x_j = 1$ and $v_e + w_{\fvec c(e)}=
1$.

Finally, we translate the system $\struct T: C \fvec x
= \fvec 1$ over~$\Zm m$ into an equivalent system over~$\Zm m$ in
which
all coefficients are either 0 or 1.
For each variable $v$ in $\struct
T$, the system has the $m$ distinct variables $v_0,
\dots, v_{m-1}$ together with equations $v_i = v_j$ for $i\not= j$.
By replacing all terms $rv$ by $\sum_{1 \leq i \leq r} v_i$ we obtain
an equivalent system. However, in order to establish our original
claim we need to express the auxiliary equations of the form $v_i =
v_j$ as a set of equations whose constant terms are equal to 1. To
achieve this, we introduce a new variable $v_j^-$ for each $v_j$,
and the equation $v_j + v_j^- + w_1 = 1$. Finally, we
rewrite $v_i = v_j$ as $v_i + v_j^- + w_1 = 1$. The
resulting system is equivalent to $\struct T$ and has the desired
form.

Since the ring $R$, and hence also the ring $\Zm m$, is fixed, it can be
seen that all the reductions we have outlined above can be formalised as
quantifier-free reductions.
\end{proof}

\begin{cor}
\label{corollary_reduction-solvability-integers} 
$\rtsqr{R} = \rtsqr{\Zm m}$, $\rtsfo{R} = \rtsfo{\Zm m}$ and
 $\rtstfo{R} = \rtstfo{\Zm m}$. 
\end{cor}

\noindent
It is a basic fact from linear algebra that solvability of a linear
equation system $A \cdot \fvec x = \fvec b$ is invariant under
applying elementary row and column operations to the augmented
coefficient matrix $(A \sep \fvec b)$. Over fields, this insight
forms the basis for the method of Gaussian elimination, which transforms
the augmented coefficient matrix of a linear equation system into row
echelon form. 
Over the integers, a generalisation of this method can be
used to transform the coefficient matrix into Hermite normal form. The
following lemma shows that a similar normal form exists over chain
rings. The proof 
uses the fact
that in a chain ring~$R$, 
divisibility 
is a total preorder.

\begin{lem}[Hermite normal form]
\label{lemma_hermite-normal-form} 
For every $k \times \ell$-matrix $A$ over a chain ring $R$, there
exists an invertible $k \times k$-matrix $S$ and an $\ell\times
\ell$-permutation matrix $T$ such that
\begin{align*}
	S A T = 
	\begin{pmatrix}
		Q \\
		\cline{2-2}
		\fvec 0
	\end{pmatrix}
	&\;\;\;\text{ with }\;\;\;
	Q = 
	\begin{pmatrix}
		a_{11} 	& 	\cdots	& \star \\
		0		&	\ddots	& \vdots \\
		 \fvec 0	& 	0		& a_{kk}
\\			
	\end{pmatrix},
\end{align*}

\noindent 
where $a_{11} \mid a_{22} \mid  a_{33}  \mid  \cdots \mid
a_{kk}$ and for all $1 \leq i, j \leq k$ it holds that
$a_{ii} \mid a_{ij}$.
\end{lem} 
\begin{proof}
  If $R$ is not a field, fix an element $\pi \in R\setminus \units R$
which generates the
maximal ideal $m = \pi R$ in $R$. Then, every element of $R$ can be
represented in the form $\pi^n u$ where $n\geq 0$ and $u \in \units R$. It
follows that for all elements $r, s \in R$ we have $r \divides s$ or
$s \divides r$. 

Now, consider the following procedure: In the
remaining $k \times \ell$-matrix $A'$, choose an entry $r \in R$ which is
minimal with respect to divisibility and use row and column
permutations to obtain an equivalent $k \times \ell$-matrix $A^\prime$
which has $r$ in the upper left corner, i.e.\ $A^\prime(1,1) = r$.
Then, use the first row to eliminate all other entries in the first
column. After this transformation, the element $r$ still divides every
entry in the resulting matrix, since all of its entries are linear
combinations of entries of $A^\prime$. Proceed in the same way with the
$(k-1) \times
(\ell -1)$-submatrix that results by deleting the first row and
column from $A^\prime$.
\end{proof}

\noindent
Now we are ready to prove the closure of $\rtsqr{R}$ and $\rtsfo{R}$
under first-order operations for the case that $R$ has prime-power
characteristic, i.e.\ for the case that $\Zm{m}$ is a chain ring.
First of all, it can be seen that
conjunction and universal quantification can be handled easily by
combining independent subsystems into a single system. 
Assume for example, that we have given two
quantifier-free interpretations 
${\mathcal{I}}=(\delta_I(\tup x),\phi_A(\tup x,\tup y))$ and
$\mathcal{J}=(\delta_J(\tup x),\psi_A(\tup x,\tup y))$ of 
$\vocLinEqRingFixed{\Zm{2}}$-structures in $\tau$-structures, where we
assume the normal form established in Lemma~\ref{lemmanormalformles}. In
order to obtain a quantifier-free interpretation
$\mathcal{I}\cap\mathcal{J}$ of $\vocLinEqRingFixed{\Zm{2}}$-structures in
$\tau$-structures such that
$(\mathcal{I}\cap\mathcal{J})(\struct A)\in\clsles{\Zm{2}}$ if, and only
if,
$\mathcal{I}(\struct A) \in \clsles{\Zm{2}}$ and $\mathcal{J}(\struct A)
\in \clsles{\Zm{2}}$, we just set
$\mathcal{I}\cap\mathcal{J} := ( \delta(x'x''\tup x), \vartheta_A(x'x''\tup
x, y'y'' \tup y))$ where
\begin{itemize}
 \item $\delta(x'x''\tup x)= ( x'=x'' \wedge
\delta_I(\tup x))\vee (x'\not=x'' \wedge \delta_J(\tup x))$, and 
 \item $\vartheta_A(x'x''\tup x, y'y'' \tup y) = (x'=x''\wedge y'=y''
\wedge \phi_A(\tup x,\tup y))\vee (x'\not=x'' \wedge y'\not=y'' \wedge
\psi_A(\tup x, \tup y)).$
\end{itemize}
To see that the resulting system is equivalent, the reader should recall
that the duplication of equations and variables does not effect the
solvability of linear equation systems.

Thus, the only
non-trivial part of the proof is to establish closure under
complementation. To do this, we describe an appropriate reduction that
translates from non-solvability to solvability of linear
systems. For this step we make use of the fact that $R$ has
prime-power characteristic (which was not necessary for obtaining the
closure under conjunction and universal quantification).

First of all, we consider the case where $R$ has
characteristic $m=p$ for a prime~$p$. In this case we know that
$\rtsqr{R} = \rtsqr{\Zm p}$ and $\rtsfo{R} = \rtsfo{\Zm p}$ by
Corollary~\ref{corollary_reduction-solvability-integers}, where $\Zm
p$ is a finite field. Over fields, the method of Gaussian elimination
guarantees that a linear equation system $(A, \fvec b)$ is not
solvable if, and only if,  for some vector $\fvec x$ we have $\fvec x
\cdot (A \sep \fvec b) = (0, \dots, 0, 1)$. In other words, the vector
$\fvec b$ is not in the column span of $A$ if, and only if, the vector
$(0, \dots, 0, 1)$ is in the row span of $(A \sep \fvec b)$. This
shows that $(A \sep \fvec b)$ is not solvable if, and only if, the
system $( (A \sep \fvec b)^T , (0,\dots, 0,1)^T)$ is solvable.  
In other words, over fields this
reasoning translates the question of non-solvability to the question of 
solvability. 
In the proof of the next lemma, we generalise this approach to chain 
rings, which enables us to translate from non-solvability to solvability 
over all rings of prime-power characteristic.

\begin{lem}[Non-solvability over chain rings]
\label{lemma_criterion_non_solvability}
Let $(A,\fvec b)$ be a linear equation system over a chain ring $R$
with maximal ideal $\pi R$ and let $n$ be the nilpotency of $\pi$.
Then $(A,\fvec b)$ is not solvable over $R$ if, and only if, there is
a vector $\fvec x$ such that $\fvec x \cdot (A \sep \fvec b) = (0,
\dots, 0, \pi^{n-1})$.
\end{lem}

\begin{proof}
Of course, if such a vector $\fvec x$ exists, then $(A,\fvec b)$  
is not solvable. On the other hand, if no such $\fvec x$ exists,
then we apply Lemma~\ref{lemma_hermite-normal-form} to transform the
augmented matrix $(A \sep \fvec b)$ into Hermite normal form $(A' \sep
\fvec b')$ with respect to $A$ (that is, $A' = S A T$ as in
Lemma~\ref{lemma_hermite-normal-form} and $\fvec b' = S \fvec b$). We
 claim that for every row index $i$, the diagonal entry $a_{ii}$ in
the transformed coefficient matrix $A'$ divides the $i$-th entry of
the transformed target vector $\fvec b'$. Towards a contradiction,
suppose that there is some $a_{ii}$ not
dividing $\fvec b'_i$. Then $a_{ii}$ is a non-unit in $R$ and 
can  be written as $a_{ii} = u \pi^t$ for some unit $u$ and
$t \geq 1$. 

By Lemma~\ref{lemma_hermite-normal-form}, it holds
that $a_{ii}$ divides every entry in the $i$-th row of $A'$ and thus
we can multiply the $i$-th row of the augmented matrix $(A' \sep \fvec
b')$ by an appropriate non-unit to obtain a vector of the form
$(0, \dots, 0, \pi^{n-1})$, contradicting our assumption. Hence, in
the transformed augmented coefficient matrix~$(A \sep \fvec b)$,
diagonal entries divide \emph{all} entries in the same row, which
implies solvability of $(A,\fvec b)$, since every linear equation of the
form $ax + c = d$ with $a, c, d \in R$ and $a \divides c, d$ is
clearly solvable.
\end{proof}

\noindent
Along with our previous discussion,
Lemma~\ref{lemma_criterion_non_solvability} now yields the closure of
$\rtsqr{R}$ and $\rtsfo{R}$ under complementation for all rings $R$ which
have prime-power characteristic. As noted above, it is an interesting open
question whether the reduction classes are also closed under
complementation when $R$ does not have prime characteristic. The
prototype example for studying this question is $R=\Zm{6}$.

\begin{thm}
\label{theorem_redtosolv-foop} $\rtsqr{R}$, $\rtsfo{R}$ and 
$\rtstfo{R}$ are closed under first-order operations for all finite
commutative rings $R$ that have prime-power characteristic.
\end{thm}

\subsection{Solvability over rings of prime characteristic}

From now on we assume that the commutative ring $R$ is of prime
characteristic $p$. We prove that in this case, the three reduction
classes
$\rtsqr{R}$, $\rtsfo{R}$ and $\rtstfo{R}$ coincide.
By definition, we have $\rtsqr{R} \subseteq \rtsfo{R}
\subseteq \rtstfo{R}$.
  Also, since we know that solvability over $R$
can be reduced to solvability over $\Zm p$ (Corollary~
\ref{corollary_reduction-solvability-integers}), it suffices for our
proof to show that $\rtsqr{\Zm p} \supseteq \rtstfo{\Zm p}$.
Furthermore, by Theorem \ref{theorem_redtosolv-foop} it follows that
$\rtsqr{\Zm p}$ is closed under first-order operations, so it only
remains to prove closure under oracle queries. 

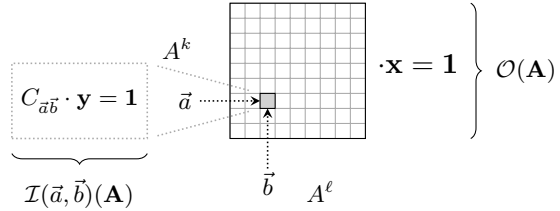
\begin{wrapfigure}{r}{0.54\textwidth}
  	\vspace{-10pt}
	\begin{center} 
		\input{./figures/solvability-nesting}
	\end{center}
	\caption{Each entry $(\tup a, \tup b)$ of the coefficient 
 matrix of the outer linear equation system $\mathcal O(\struct A)$ is
 determined by the corresponding inner linear system $C_{\tup a\tup b}
\cdot \fvec y = \fvec 1$ described by $\mathcal I(\tup a, \tup b)(\struct
A)$: this entry is~1 if $\mathcal I(\tup a, \tup b)(\struct A)$ is solvable
and 0 otherwise.}
	\label{fig_solvability-nesting}
  	\vspace{-10pt}
\end{wrapfigure}
\noindent 
More specifically,
it can be seen that proving closure under oracle queries amounts to showing
that the nesting of solvability queries can be
reduced to the solvability of a single linear equation system. In order to
formalise
this, let $\mathcal{I}(\bx,\by)$ be a
quantifier-free interpretation of $\vocLinEqRingFixed{\Zm p}$ in
$\sigma$ with parameters $\bx, \by$ of length $k$ and $\ell$,
respectively. We extend the signature $\sigma$ to $\sigma_X \defeq
\sigma \cup \{ X \}$ and restrict our attention to those
$\sigma_X$-structures~$\struct{A}$ (with domain $A$) where the
relation symbol $X$ is interpreted as $X^{\struct{A}} = \lbrace
(\ba,\bb) \in A^{k \times \ell} \sep \mathcal{I}(\ba,\bb)(\struct{A}) \in
\clsles{\Zm p} \rbrace$. 

\medskip
Then it remains to show that for any
quantifier-free interpretation $\mathcal{O}$ of
$\vocLinEqRingFixed{\Zm p}$ in~$\sigma_X$, there exists a quantifier-free
interpretation of $\vocLinEqRingFixed{\Zm p}$ in $\sigma$ that
describes linear equation systems equivalent to $\mathcal{O}$.

Hereafter, for any $\sigma_X$-structure $\struct A$ and tuples $\tup
a$ and $\tup b$, we will refer to $\mathcal O(\struct A)$ as an
``outer'' linear equation system and refer to $\mathcal I(\tup a, \tup
b)(\struct A)$ as an ``inner'' linear equation system. By applying
Lemma~\ref{lemmanormalformles} and
Theorem~\ref{theorem_redtosolv-foop}, it is sufficient to consider the
case where for $\sigma_X$-structures $\struct{A}$,
$\mathcal{O}(\struct{A})$ describes a linear system $(M,
\fvec 1)$, where $M$ is the $\{ 0,1 \}$-matrix of the relation
$X^{\struct{A}}$. For an illustration of this setup, 
see Figure~\ref{fig_solvability-nesting}.


\begin{thm}[Closure under oracle queries]
For $\mathcal{I}$, $\mathcal{O}$ as above, there exists a quantifier-free 
interpretation $\mathcal{K}$ of $\vocLinEqRingFixed{\Zm p}$ in $\sigma$
such that for all $\sigma_X$-structures $\struct{A}$ it holds that
$\mathcal{O}(\struct{A}) \in \clsles{\Zm p}$ if, and only if,
$\mathcal{K}(\struct{A}) \in \clsles{\Zm p}$.
\end{thm}

\begin{proof} 
For a $\sigma$-structure $\struct{A}$, let $M_o$ denote the 
$\{ 0,1 \}$-coefficient matrix of the outer linear equation system
$\mathcal{O}(\struct{A})$. Then for $(\ba, \bb) \in A^{k \times \ell}$ 
we have $M_o(\ba,\bb)=1$ if, and only if, the inner linear system
$\mathcal{I}(\ba, \bb)(\struct{A})$ is solvable. If we explicitly denote
the set of variables of the outer linear system $\mathcal{O}(\struct{A})$
by $\lbrace v_{\bb} \sep \bb \in A^\ell \rbrace$, then we can express the
equations of $\mathcal{O}(\struct{A})$ as $\sum_{\bb \in A^\ell} M_o(\ba,
\bb) \cdot v_\bb = 1$, for $\ba \in A^k$.

For the new linear equation system $\mathcal{K}(\struct{A})$ we take
$\lbrace v_{\ba, \bb} \sep (\ba, \bb) \in A^{k \times \ell} \rbrace$ as
the set of variables, and we include the equations $\sum_{\bb \in
A} v_{\ba, \bb} = 1$ for all $\ba \in A^k$. In what follows, our aim is to
extend $\mathcal{K}(\struct{A})$ by additional equations so that for every
solution to $\mathcal{K}(\struct{A})$, there are values~$v_\bb \in \Zm p$
such that for $\ba \in A^k$ it holds that $v_{\ba, \bb} = M_o(\ba, \bb)
\cdot v_{\bb}$. Assuming this to be true, it is immediate that
$\mathcal{O}(\struct{A})$ is solvable if, and only if,
$\mathcal{K}(\struct{A})$ is solvable, which is what we want to show. 

\noindent
In order to enforce the condition
``$v_{\ba, \bb} = M_o(\ba, \bb) \cdot v_{\bb}$'' by \emph{linear}
equations, we need to introduce a number of auxiliary linear subsystems to
$\mathcal K(\struct A)$. The reason why we cannot express this condition
directly by a linear equation is because $M_o(\ba, \bb)$ is determined by
the solvability of the inner system $\mathcal{I}(\ba,\bb)(\struct{A})$.
Therefore, if we were to treat both the elements of $M_o(\ba, \bb)$ and the
$v_\bb$ as individual variables, then that would require us to express the
\emph{non-linear} term $M_o(\ba,\bb) \cdot v_{\bb}$.

\medskip
To solve this
problem, we introduce new linear subsystems in~$\mathcal{K}(\struct{A})$
to ensure that for all $\ba, \bb, \bc \in A$ it holds that:
\begin{align}
&\textbf{if }  v_{\ba, \bb} \not= 0 \textbf{ then } M_o(\ba, \bb) = 1;
\text{ and } \label{equation_proof-nesting-solv1} \\
&\textbf{if }  v_{\ba, \bb} \not= v_{\bc, \bb} \textbf{ then } \lbrace
M_o(\ba, \bb) , M_o(\bc, \bb) \rbrace = \{ 0,1
\rbrace.\label{equation_proof-nesting-solv2}
\end{align}

\noindent
Assuming we have expressed (\ref{equation_proof-nesting-solv1}) 
and (\ref{equation_proof-nesting-solv2}), it can be seen that solutions of
$\mathcal{K}(\struct{A})$ directly translate into solutions for
$\mathcal{O}(\struct{A})$ and vice versa. To express
(\ref{equation_proof-nesting-solv1}) we proceed as follows: for each $(\ba,
\bb) \in A^{k \times \ell}$ we add $\mathcal{I}(\ba,\bb)(\struct{A})$
as an independent linear subsystem in $\mathcal{K}(\struct{A})$ in which we
additionally extend each equation by the term $(v_{\ba,\bb} +1)$.
Now, if in a solution of~$\mathcal{K}(\struct{A})$ the variable
$v_{\ba,\bb}$ is evaluated to $0$, then the subsystem corresponding to
$\mathcal{I}(\ba,\bb)(\struct{A})$ has the trivial solution
(recall, that the constant terms of every equations are $1$). However, if a
non-zero value is assigned to $v_{\ba, \bb}$, then this value clearly is a unit 
in
$\Zm p$ and thereby a solution for $\mathcal{K}(\struct{A})$ necessarily
contains a solution of the subsystem $\mathcal{I}(\ba,\bb)(\struct{A})$;
 that is, we have $M_o(\ba, \bb) = 1$.

 \smallskip
For (\ref{equation_proof-nesting-solv2}) we follow a similar approach. 
For fixed tuples $\ba$, $\bb$ and $\bc$, the condition on the right-hand
side of (\ref{equation_proof-nesting-solv2}) is a simple Boolean
combination of solvability queries. Hence, by Theorem
\ref{theorem_redtosolv-foop}, this combination can be expressed by a single
linear equation system. Again we embed the respective linear equation
system as a subsystem in $\mathcal{K}(\struct{A})$ where we add to each of
its equations the term $(1 + v_{\ba, \bb} - v_{\bc, \bb})$. With the same
reasoning as above we conclude that this imposes the constraint
(\ref{equation_proof-nesting-solv2}) on the variables $v_{\ba,\bb}$ and
$v_{\bc, \bb}$, which concludes the proof.
\end{proof}

\begin{cor}
\label{corollary_redtosolv-primechar}
If $R$ has prime characteristic, then $\rtsqr{R}=\rtsfo{R}=\rtstfo{R}$.  
\end{cor}

\noindent
As explained above, our results have some important
consequences. For a prime $p$, let us denote by $\foslvp$
first-order logic extended by quantifiers deciding
solvability over $\Zm{p}$. 
Corresponding extensions of first-order logic by rank operators
over prime fields $(\forkp)$ were studied by
Dawar~\etal~\cite{dawar09logics}. Their results imply that $\foslvp =
\forkp$ over ordered structures, and that both logics have a strong
normal form over ordered structures, i.e.\ that every formula is
equivalent to a formula with only one application of a solvability or
rank operator, respectively~\cite{pakusa10thesis}. 
Corollary~\ref{corollary_redtosolv-primechar} allows
us to generalise the latter result for $\foslvp$ to the class of all finite
structures.

\begin{cor}
\label{cor_redtosolv-normalformfo}
Every formula $\phi \in \foslvp$ is equivalent to a formula with a single
application of a solvability quantifier to a quantifier-free formula.
\end{cor}

%
%
\section{Linear algebra in fixed-point logic with
counting} 
\label{sec_linear-algebra-fpc}

\renewcommand{\thefootnote}{\arabic{footnote}}

In this section we present further applications of the techniques we
developed during our study of the solvability problem 
that we pursued so far. Specifically, we turn our
interest to the elements of linear algebra over finite commutative
rings that can be expressed in fixed-point logic with counting. 
While the solvability problem over commutative rings is not definable in
fixed-point logic with counting, we show here that various other natural
matrix problems and properties, such as matrix inverse and matrix
determinant, can be formalised in \fpc over certain classes of
rings. 

To this end, we first apply the definable structure theory that
we established in~\S\ref{sec_structure-of-finite-rings} to show that many
linear-algebraic
problems over commutative rings can be reduced to the
respective problems over local rings. Next, we consider basic
linear-algebraic queries over local rings,
such as multiplication and matrix inversion, for the case where the
local ring comes with a built-in linear ordering (that is, where the
matrix is encoded as a finite $\vocMatrixOrd$-structure). By
Lemma~\ref{lemma_ordering-local-pir} it follows that all
of these definability results hold for matrices over $k$-generated
local rings (for a fixed $k$) as well, since we can define in $\fpc$
a linear order in such rings.
In particular, all of our results on $\vocMatrixOrd$-structures apply
to matrices over chain rings (that is $1$-generated, or principal ideal,
local rings), which
include the class of so called \emph{Galois rings}.
In the final part of this section we study matrices over unordered
Galois rings specifically. Here our main result is that there are
$\FPC$-formulas
that define the characteristic polynomial and the determinant of square
matrices over Galois rings, which extends the results of
Dawar~\etal~\cite{dawar09logics} concerning definability of the determinant
of matrices over finite fields. 


\subsection{Basic linear-algebraic operations over commutative rings}

To begin with, we need to fix our encoding of matrices over a
commutative ring $R$ in terms of a finite relational structure. For this
we
proceed as we did for linear equation systems
in~\S\ref{sec_background}. As before, for non-empty sets $I$ and $J$, an
$I\times J$-matrix $A$ over a commutative ring $R$ is a mapping
$A: I\times J \to R$. We set $\vocMatrix :=
( R, M , \vocRing)$, where $R$ is a unary, and $M$ is a
ternary relation symbol, and we understand each $\vocMatrix$-structure
$\struct A$
as encoding a matrix over a commutative ring 
(given that $\struct A$ satisfies some basic properties which make this
identification well-defined, e.g.\ that $R^{\struct A}$ forms a commutative
ring and that the projection of $M^{\struct A}$ onto the third component is
a subset of $R^\struct{A}$).
If we want to encode \emph{pairs} of matrices over the same commutative
ring $R$
as a finite structure, we use the extended vocabulary $\vocMatrixPair :=
( N , \vocMatrix)$, where $N$ is an additional ternary
relation symbol.
Moreover, we consider a representation of matrices as structures 
in vocabulary $\vocMatrixOrd \defeq (\vocMatrix, \leqslant )$ where it is 
assumed that
$\leqslant^\struct{A}$ is a linear ordering on the set of ring elements
$R^\struct{A}$. Similarly, structures of vocabulary $\vocMatrixPairOrd
\defeq (\vocMatrixPair , \leqslant )$ are used to encode
pairs of matrices over the ring~$R^{\struct A}$ on which
$\leqslant^\struct{A}$ is a linear ordering.

We first recall from \S\ref{sec_structure-of-finite-rings} that every
(finite) commutative ring $R$ can be expressed as a direct sum 
of local rings, i.e.\ $R=\bigoplus_{e \in \ringBase R} e \cdot R$ where
all principal ideals $eR$ are local rings.
It follows that every $I\times J$-matrix $A$ over $R$ can be uniquely
decomposed into a set of matrices $\{ A_e : e \in \ringBase R\}$, where $A
= \bigoplus_{e \in \ringBase R} A_e$ and where $A_e$ is an $I\times
J$-matrix over the local ring~$e \cdot R$. Moreover, following
Lemma~\ref{lemma_projection-local-ring}, this set of matrices can be
defined
in first-order logic. Stated more precisely,
Lemma~\ref{lemma_projection-local-ring} implies the existence of a
parameterised first-order interpretation $\Theta(x)$ of $\vocMatrix$ in
$\vocMatrix$ such that for any $\vocMatrix$-structure $\struct A$, which
encodes an $I\times J$-matrix $A$ over the commutative ring $R$, and all $e
\in \mathcal{B}(\struct R)$, it holds that $\Theta[e](\struct{A})$ encodes
the projection of $A$ onto the local ring $e \cdot R$, i.e.\ the $I \times
J$-matrix $A_e$ over the local ring $e \cdot R$. Since the ring base $\ringBase R$
of $R$ is also definable in first-order logic
(Lemma~\ref{lemma_ring-idempotent-base}), this allows us to reduce most
natural problems in linear algebra from arbitrary commutative rings to
local rings. In particular, we are interested in the following
linear-algebraic queries.

\begin{prop}[Local ring reductions]
\label{proposition_reduction-to-local-rings}
For each of the following problems over commutative rings, there is a
first-order Turing reduction to the respective query over local rings:
\begin{enumerate}
	\item Deciding solvability of linear equation systems (cf.\
\S\ref{sec_structure-of-finite-rings}).
    
	\item Defining basic matrix arithmetic, i.e.\ matrix addition and
matrix multiplication.

       \item Deciding whether a square matrix is invertible (and defining
its inverse in this case).

	\item Defining the characteristic polynomial and the determinant of
a square matrix.

	
	
\end{enumerate}
\end{prop}

\noindent
It was shown in~\cite{blass02polynomial} that over finite fields, 
the class of invertible matrices can be defined in $\FPC$ and that there is
an $\FPC$-definable interpretation that associates each invertible matrix
with its inverse. In what follows we show that the same holds when we
consider matrices over ordered local rings. Our proof 
follows the approach taken by Blass~\etal in~\cite{blass02polynomial}. As a
first step, we show that exponentiation of matrices can be defined in
$\FPC$, even if the exponent is given as a binary string of polynomial
length. We then show that for each set $I$, there is
a formula of $\FPC$ that defines the number of invertible $I \times I$
matrices over a local ring $R$. Finally, combining the two results with the
fact that the set of all invertible $I \times I$ matrices over $R$ forms a
group under multiplication, we conclude that the inverse to any invertible
$I \times I$ matrix over $R$ can be defined in $\FPC$.

\makebreak

\noindent
For the first step, to show that powers of matrices over ordered 
local rings can be interpreted in $\FPC$, we need the following lemma on
matrix multiplication. Recall that addition and multiplication of unordered
matrices is defined in exactly the same way as for ordered matrices, except
that we now have to ensure that the index sets of the two matrices, and not
just their dimension, are matching. 
\begin{lem}[Matrix multiplication]

\label{lemma_matrix-multiplication} 
 There is an $\FPC$-interpretation $\Theta$ of $\vocMatrixOrd$ in
 $\vocMatrixPairOrd$ such that for all $\vocMatrixPairOrd$-structures
 $\struct P$, which encode an $I \times K$-matrix $A$ and a
 $K \times J$-matrix $B$ over a commutative ring
 $R$ with a linear order, $\Theta(\struct{P})$
encodes the $I \times J$-matrix $A \cdot B$.
\end{lem}

\begin{proof} 
We reuse an idea of Blass~\etal~\cite{blass02polynomial}. For $i \in I$
and $j \in J$ we have
\[ 
	(A\cdot B)(i,j) 
	= \sum_{k \in K} A(i,k) \cdot B(k,j) 
	= \sum_{r \in R} r \cdot \bigl[\# {k \in K} \bigl(  A(i,k)
\cdot B (k,j) = r \bigr) \bigr].
\]

\noindent
In other words, instead of individually summing up all products $A(i,k)
\cdot B(k,j)$ for indices $k \in K$, which would require a linear order
on $K$, we use the counting mechanism of $\FPC$ to obtain the multiplicities of
how often a specific ring element $r \in R$ appears in this sum. The
entries of $A \cdot B$ can then easily be obtained by taking the sum over all
ring elements $r \in R$ weighted by the respective multiplicities (the
right-hand expression). Since the ring $R$ is ordered, this expression can
easily be defined in $\FPC$. 
\end{proof}

\noindent
In~\cite{blass02polynomial}, Blass~\etal showed that exponentiation 
of square matrices over a finite field can be expressed in $\FPC$ (using
the fact that matrix multiplication is in $\FPC$). Moreover, by using the
method of repeated squaring, they show that
exponentiation can be expressed even when the value of the exponent is
given as a binary string of length polynomial in the size of the
input structure.

\begin{defi}
Let $\eta(\upsilon)$ be an $\FPC$-formula with a free number variable 
$\upsilon$ and let $t$ be a number term of $\FPC$ with no free variables.
For a structure $\struct B$ we let
\[
	(\eta(\upsilon), t)^\struct{B} \defeq \{ i \sep 0 \leq i \leq
t^\struct{B} \text{ and } \struct B \models \eta[i] \}.
\]

\noindent
We identify the set $(\eta(\upsilon), t)^\struct{B}$ with 
the integer $m = \sum_{i \in (\eta, t)^\struct{B}} 2^i$, i.e.\ the tuple
$(\eta(\upsilon), t)$ defines in the structure $\struct B$ the
$t^\struct{B}$-bit binary representation of $m$. \defnend
\end{defi}



\noindent
Given that the product of two matrices can be defined in
 $\FPC$ (Lemma~\ref{lemma_matrix-multiplication}), it is not hard to see
that the repeated squaring approach outlined in~\cite{blass02polynomial}
for describing matrix exponentiation over finite fields also works for
matrices over commutative rings with a linear order. We state this more
formally as follows.


\begin{lem}[Matrix exponentiation]
\label{lemma_matrix-exponentiation} 
For each pair $(\eta(\upsilon), t)$, where $\eta(\upsilon)$ is an
 $\FPC$-formula with a free number variable $\upsilon$ and $t$ is a number
term, there is an $\FPC$-interpretation $\Theta_{(\eta(\upsilon),  t)} $ of
$\vocMatrixOrd$ in $\vocMatrixOrd$ such that for all $I \times I$ matrices
$\struct A$ over a local ring $\struct R$ with a linear order, the
structure $\Theta_{(\eta(\upsilon), t)}^\struct{A}$ encodes the $I \times
I$ matrix $\struct A^{n}$, where $n = (\eta(\upsilon), t)^\struct{A}$. \qed
\end{lem}

\noindent
Recall that the set of invertible $I \times I$-matrices over a commutative
ring $R$ forms a group under matrix multiplication, which is
known as the general linear group (written $\genl{I}{R}$). Hence, if we let
$\ell := \card{\genl{I}{R}}$ then an $I\times I$-matrix $A$ over $R$ is
invertible if, and only if, $A^\ell = \IdMatrix$. The following lemma shows
that the cardinality of the general linear group for local rings $R$
can be defined in $\FPC$. 

\label{lemma_cardinality-gen-lin-group-definable}
\begin{lem}[Cardinality of $\genl{I}{R}$]
There is an $\FPC[\vocMatrix]$-formula $\eta(\upsilon)$, 
with a free number variable $\upsilon$, and a number term $t$ (without free
variables) in
$\FPC[\vocMatrix]$, such that for any 
$I\times I$-matrix $\struct A$ over a local ring $R$, it holds that
$(\eta(\upsilon), t)^\struct{A} = \card{\genl{I}{\struct R}}$.
\end{lem}
\begin{proof}
Let $R$ be a local ring with maximal ideal $m$ and let $I$  be a
finite set of size $n > 0$. We denote the field $R / m$ by $F$
and its cardinality by $q \defeq \card{F}$. Then we have
\[ 
	\card{\genl{I}{F}} = q^{n^2} \cdot \prod_{i=0}^{n-1} (1-q^{i-n}) 
	\quad\text{ and }\quad
	\card{\genl{I}{R}} = \card{R}^{n^2} \cdot \prod_{i=0}^{n-1}
(1-q^{i-n}).
\]

\noindent
Indeed, the first equation is easy to verify: an $I \times I$ matrix over
$F$ is invertible if, and only if, its columns are linearly independent.
Each set of $i$ linearly independent columns generate $q^i$ different
vectors in $F^I$. Thus, if we have already fixed $i$ linearly independent
columns, there remain $(q^{n} - q^i)$ vectors in $F^I$ which can be used
extend to this set to a set of $i+1$ linearly independent columns. This
counting argument shows that
\begin{align*}
	\card{\genl{I}{F}} &= (q^{n} -1 ) \cdot (q^{n} - q ) \cdot (q^{n}
-q^2 ) \cdots (q^{n} -q^{n-1} ) =  \prod_{i=0}^{n-1} (q^{n} - q^i ).
\end{align*}

\noindent
For the second equation, we let $\pi$ denote the canonical group
epimorphism $\pi: \genl{I}{R} \ra \genl{I}{F}$. It is easy to see that
$\card{\mathrm{ker}(\pi)} = \card{m}^{n^2}$. The homomorphism
theorem thus implies that $\card{\genl{I}{\struct R}} =  {\card{m}^{n^2}}
\card{\genl{I}{F}}$ which yields the claim since $q \cdot
\card{m} = \card{R}$.

\medskip
By the above claim we have an exact expression
for the cardinality of $\card{\genl{I}{R}}$ for any non-empty set $I$ and
any finite local ring $R$. It is straightforward to verify that the binary
encoding of this expression can be formalised by a formula and a number
term of fixed-point logic with counting.
\end{proof}

\begin{thm}[Matrix inverse]
\hfill
\begin{enumerate}
 \item There is an $\FPC$-interpretation $\Theta$ of $\vocMatrixOrd$ in
$\vocMatrixOrd$ such that for any $\vocMatrixOrd$-structure~$\struct A$
encoding an $I \times I$-matrix $A$ over a 
commutative ring $R$ with linear order, 
$\Theta(\struct{A})$ encodes an $I\times I$-matrix $B$ such
that $AB = \IdMatrix$, if $A$ is invertible.
\item For every $k \geq 1$, there is an $\FPC$-interpretation $\Theta$ of
$\vocMatrix$ in $\vocMatrix$ such that for any $\vocMatrix$-structure
$\struct A$ which encodes an $I \times I$-matrix $A$ over a
commutative ring $R$ that splits into a direct sum of $k$-generated local
rings, it holds that $\Theta(\struct{A})$ encodes an $I\times I$-matrix $B$
over $R$ such that $B = \inv{A}$, if $A$ is invertible, and
$B=\ZeroMatrix$, otherwise.
\end{enumerate}
\end{thm} 
\begin{proof}
For the first claim, we combine the arguments outlined
above. For the second claim, we additionally apply
Lemma~\ref{lemma_ordering-local-pir} to obtain a linear order on the local
summands of $R$.
\end{proof}





\subsection{Characteristic polynomial over Galois rings}

In~\cite{blass02polynomial}, Blass~\etal
established
that the problem of deciding singularity of a square
matrix over $\GF{2}$ can be expressed in $\FPC$, as we already explained above. Recall that a matrix $A$ is singular
over a field if, and only if, its determinant $\det(A)$ is zero. The result
of Blass~\etal therefore implies that the determinant of a matrix over
$\GF{2}$ can be expressed in $\FPC$ by testing for singularity. This result
was generalised by Dawar~\etal~\cite{dawar09logics}, who showed that over
any finite field (as well as over $\Z$ and $\Q$) the characteristic
polynomial (and thereby, the determinant) of a matrix can be defined in
$\FPC$ (for full details,
see~\cite{holm10thesis}).
Pakusa~\cite{pakusa10thesis} observed that the
same approach works for
the definability of the determinant and characteristic polynomial of
matrices over prime rings $\Zm{p^n}$. 
Recall that for an $I\times I$-matrix $A$ over a commutative ring~$R$, the 
characteristic polynomial $\chi_A \in R[X]$ of $A$ is defined as $\chi_A = 
\det(X E_I - A)$, where $E_I$ denotes the $I\times I$-identity matrix over $R$.

Here we show that the characteristic polynomial of 
matrices over any Galois ring can also be defined in $\FPC$. A finite
commutative ring $R$ is called a \emph{Galois ring} if it is a Galois
(ring) extension of the ring $\Zm{p^n}$ for a prime $p$ and $n \geq
1$. As we will only work with the following equivalent
characterisation of Galois rings we omit the definition of 
Galois ring extensions (for details, we refer to
\cite{bini02finite,mcdonald1974finite}).

\begin{defi}
A \emph{Galois ring} is a finite commutative ring $R$ which is 
isomorphic to a quotient ring $\Zm{p^n}[X] / (f(X))$, where $f(X)$ is a
monic irreducible polynomial of degree $r$ in $\Zm{p^n}[X]$ whose image
under the reduction map $\mu: \Zm{p^n} \rightarrow \Zm{p^n} / (p \cdot
\Zm{p^n}) \isom \GF{p}$ is irreducible.
Such a polynomial is
called a \emph{Galois polynomial} of degree $r$ over $\Zm{p^n}$.\defnend 
\end{defi}

\noindent
We summarise some useful facts about Galois rings. For every
ring~$\Zm{p^n}$ and every $r \geq 1$, there is a \emph{unique}
Galois extension of degree $r$ over $\Zm{p^n}$, which we denote by
$\GR{p^n}{r}$. 
Moreover, any Galois ring is a chain ring, and thus we
can use Lemma~\ref{lemma_ordering-local-pir} to obtain an \FPC-definable
linear order on such rings. As Galois rings include all finite fields
and all prime rings, the following theorem gives a  generalisation of
all known results concerning the logical complexity of
the characteristic polynomial and determinant from~\cite{holm10thesis,dawar09logics,blass02polynomial}. 

\begin{thm}[Characteristic polynomial]
\label{theorem_charpoly-galois-rings}
There are $\FPC$-formulas $\theta_{\text{det}}(z)$ and 
$\theta_{\text{char}}(z,\upsilon)$, where $z$ is an element variable and
$\upsilon$ is a number variable, such that for any 
$\vocMatrix$-structure $\struct A$ which encodes an $I\times I$-matrix $A$
over a Galois ring $R$ we have:
\begin{itemize}

    \item $\struct A \models \theta_{\text{det}}[d]$ if, and only if, the
determinant of $A$ over $R$ is $d \in R$;

    \item $\struct A \models \theta_{\text{char}}[d,k]$ if, and only if,
the coefficient of $x^k$ in the characteristic polynomial $\chi_{\struct
A}(x)$ of $A$ over $R$ is $d \in R$.
\end{itemize}
\end{thm}

\noindent
Before we prove this theorem, we need some technical results. First of all,
we fix an encoding of polynomials by number terms of $\FPC$, as follows.

\begin{defi}[Encoding polynomials]
Let $\pi(\upsilon)$ be a number term (of $\FPC$) in signature $\tau$, where $\upsilon$ is
a number variable. Given a $\tau$-structure $\struct A$ and an integer
$m$, we write $\function{poly}_{X}(\pi, \struct A, m)$ to denote the
integer polynomial $a_{m}X^m + \dots + a_1X + a_0$, where $a_i =
\pi[i]^{\struct{A}}$ for each $i \in [m]$. \defnend
\end{defi}

\makebreak

\noindent
By the definition of Galois rings, we know that each $\GR{p^n}{r}$ is 
isomorphic to a quotient ring $\Zm{p^n}[X] / (f(X))$, where $f(X)$ is
a Galois polynomial of degree $r$ over $\Zm{p^n}$.
The following lemma shows that we can define this isomorphism explicitly 
in $\FPC$ over each Galois ring.

\begin{lem}[Representation of Galois rings]
\label{lemma_definable-isomorphism-galois-rings}
Let $\upsilon$ be a number variable and $x, y, z$ be element variables.
There is an $\FPC$-formula $\phi(x,y)$ and $\FPC$-number terms $\eta(\upsilon; x,y,z)$
and $\pi(\upsilon; x,y)$ in vocabulary $\vocRing$ such that for any Galois ring $R \isom \GR{p^n}{r}$ and
all pairs $(\alpha, \beta) \in \phi(x,y)^{R}$, it holds that the map $\iota\colon R \rightarrow \Zm{p^n}[X]$ given by 
\[
	\iota: g \mapsto
	\function{poly}_{X}(\eta(\upsilon; \alpha/x,\beta/y, g/z),  R, r) \bmod{p^n},
\]

\noindent
is a ring isomorphism $\iota\colon R \isom \Zm{p^n}[X]/(f(X))$, where
$f(X)$ is
the Galois polynomial of degree~$r$ over $\Zm{p^n}$ given by $f(X) \defeq
\function{poly}_{X}(\formula{\pi}(\upsilon; \alpha/x, \beta/y), R, r) \bmod{p^n}$.
\end{lem}

\begin{proof}
We recall the construction of $\GR{p^n}{r}$ described by 
Bini and Flamini~\cite{bini02finite}. Recall that the residue field
$\GF{p^r}$ of $\GR{p^n}{r}$ is isomorphic to $\Zm{p}[X]/(g(X))$ where
$g(X)$ is a monic polynomial of degree $r$ over $\Zm{p}$. More
specifically, $g(X)$ is the least monic polynomial in $\Zm{p}[X]$ such that
$g(\alpha) = 0$, where $\alpha \in \GF{p^r}$ is a primitive element of
$\GF{p^r}$; that is, a generator of the cyclic group $\units{\GF{p^r}}$.
Let one such $g(X)$ be fixed. Then $\GR{p^n}{r} \isom \Zm{p^n}[X]/(f(X))$
where $f(X)$ is a polynomial of degree $r$ in $\Zm{p^n}[X]$ such that $f(X)
\equiv g(X) \imod{p}$. Moreover, it is easy to construct the polynomial
$f(X)$ from $g(X)$, as follows. Writing $g(X) = X^r + a_{r-1}X^{r-1} +
\cdots + a_1 X + a_0$ and $f(X) = X^r + b_{r-1}X^{r-1} + \cdots + b_1 X +
b_0$, it is shown in~\cite{bini02finite} that $b_i = p^n - p + a_i$ for $i
\in [0, r-1]$. 

\smallskip

As shown by Holm~\cite[Chapter 3]{holm10thesis}, we can define both
 the set of primitive elements and their associated minimal polynomials in
$\FPC$ over finite fields. Assume we have fixed a primitive element
$\alpha$ and its minimal polynomial $g(X)$. Then, by the above, it is
straightforward to formalise $f(X)$ in $\FPC$, with the element $\alpha$ as
a parameter. Finally, if we let $\beta \in \GR{p^n}{r}$ be a root of $f(X)$ 
then the isomorphism  $\iota: \GR{p^n}{r} \isom \Zm{p^n}[X]/(f(X))$ can be
realised explicitly as $a \mapsto h_a(X)$ where $h_a(X)$ is the unique
polynomial such that $h_a(\beta) = a$, for each $a \in R$~\cite{mcdonald1974finite}; 
that is to say, the map $\iota$ is given by
\[
	\iota: a \mapsto h(X) \defiff h(\beta) = a,
\]

\noindent
and this mapping can be formalised in $\FPC$ given $\alpha$ and $\beta$ as parameters.
\end{proof}

%

\noindent
To prove Theorem~\ref{theorem_charpoly-galois-rings}, we follow the
approach\footnote[1]{This approach was originally suggested by Rossman and
documented in a note by Blass and Gurevich\cite{blass05update}.} described
in~\cite{dawar09logics} and formalise a well-known polynomial-time
algorithm by Csanky~\cite{csanky76fast} in $\FPC$. This algorithm computes
the coefficients of the characteristic polynomial (and thereby, the
determinant) of a matrix over any commutative ring of characteristic zero.
The technical details of how to express Csanky's algorithm in $\FPC$ have
been
outlined in~\cite{dawar09logics,holm10thesis,laubner11thesis} and we omit
the details here.

Since Csanky's algorithm only works in characteristic zero, it cannot be
applied directly to matrices over Galois rings. Instead, given a matrix
$A$ over a Galois ring $R \isom \GR{p^n}{r}$, we take the
following steps to ensure that $A$ is suitable for the algorithm.

Firstly, using Lemma~\ref{lemma_definable-isomorphism-galois-rings} we
define a polynomial $f(X) \in \Zm{p^n}[X]$ such that $R \isom
\Zm{p^n}[X] / (f(X))$ (and this isomorphism can be defined explicitly). Let
$F(X)$ be a polynomial over $\Z$ whose reduction modulo $p^n$ is $f(X)$.
This can be done trivially, for the coefficients of $f(X)$ are already
given by integers in the range $[0,p^n-1]$. Next we lift the matrix
$A$ to a matrix $A^\star$ over the quotient ring $S
= \Z[X]/(F(X))$: first by translating $A \mapsto \iota(A)$
according to the $\FPC$-definable isomorphism $\iota: R \isom
\Zm{p^n}[X]/(f(X))$ given by
Lemma~\ref{lemma_definable-isomorphism-galois-rings}
 and then map $\iota(A) \mapsto {A}^\star$ by lifting each $h(X)$ in
$\Zm{p^n}[X]/(f(X))$ to the polynomial $H(X)$ in $S$ whose
reduction modulo $p^n$ is $h(X)$. Finally, we apply Csanky's algorithm over
$S$ to the matrix $A^\star$ and then reduce the output
modulo $p^n$ to get the correct result. This last reduction is sound as we
have $R = S/(p^n)$. As explained
in~\cite[\S3.4.3]{holm10thesis}, Csanky's algorithm can be formalised in
$\FPC$ even when the ring elements are given explicitly as polynomials in
this way. Putting everything together, we conclude that each coefficient of
the characteristic polynomial $\chi_{A}(X)$ of $A$ can be
defined in $\FPC$. Since the determinant of $A$ is precisely the
constant term of $\chi_{A}(X)$,
Theorem~\ref{theorem_charpoly-galois-rings} now follows.

%
%
\section{Discussion}
\label{sec_discussion}

Motivated by the question of finding extensions
of \fpc to capture
larger fragments of $\Ptime$, we have analysed the
(inter-)definability of solvability problems over various classes of
algebraic domains. Similar to the notion of rank logic
\cite{dawar09logics} one can consider \emph{solvability logic}, which
is the extension of \fpc by  Lindstr\"om quantifiers that
decide solvability of linear equation systems. In this context, our
results from
\S\ref{sec_les-different-domains}~and~{\S\ref
{sec_structure-of-finite-rings}}
 can be seen to relate fragments of
solvability logic obtained by restricting quantifiers to different
algebraic domains, such as Abelian groups or commutative rings. We
have also identified many classes of algebraic structures over which
the solvability problem reduces to the very basic problem of
solvability over cyclic groups of prime-power order. This raises the
question, whether a reduction even to groups of prime order is
possible. In this case, solvability logic would turn out to be a
fragment of rank logic. On the other hand, it also remains open
whether or not the matrix rank over finite fields can be
expressed in fixed-point logic extended by solvability operators.

With respect to specific algebraic domains, we prove that $\fpc$ can
define a linear order on the class of all $k$-generated local
rings, i.e.\ on classes of local rings for which every maximal ideal can
be generated by $k$ elements, where $k$ is a fixed constant. Together
with our results from \S\ref{sec_structure-of-finite-rings}, this 
can
be used to show that all natural problems from linear algebra over
(not necessarily local) $k$-generated rings reduce to problems
over ordered rings under $\fp$-reductions. An interesting direction of
future research is to explore how far our techniques can be used to
show (non-)definability in fixed-point logic of other problems from
linear 
algebra over rings. 
 
Finally, we mention an interesting topic of related research, which is the logical study of
\emph{permutation group membership
problems} ($\text{GM}$ for short). An instance of $\GM$ consists of
a set $\Omega$, a set of generating permutations $\pi_1, \dots,
\pi_n$ on $\Omega$ and a target permutation $\pi$, and the problem is
to decide whether $\pi$ is generated by $\pi_1, \dots, \pi_n$. This
problem is known to be decidable in polynomial time (indeed it is in
$\compnc$ \cite{BaLuSe87}). We can show that all the
solvability problems we studied in this paper reduce to $\GM$
under first-order reductions (basically, an application of Cayley's
theorem). 
In particular this shows that $\GM$ is not definable in
$\fpc$. By extending fixed-point logic by a suitable operator for $\GM$ we
therefore obtain a logic which extends rank logics and in which all studied solvability problems are
definable.
This logic is worthy of further study as it
can uniformly express all problems from (linear) algebra that have
been considered so far in the context of understanding the descriptive
complexity gap between $\FPC$ and $\Ptime$.

%
%
\bibliographystyle{plain} 
\bibliography{main}

%
%

	

%
%

%
%

%
%

%
%

\end{document}

%% file: setup/document-setup.tex
\usepackage{hyperref}
\usepackage{xspace}
\usepackage{wrapfig}	

\usepackage{tikz}
\usetikzlibrary{matrix,arrows}
\usetikzlibrary{backgrounds,snakes}
\usetikzlibrary{calc}

\renewcommand{\thefootnote}{\fnsymbol{footnote}}

\newcommand{\defnend}{\hfill\ensuremath{\blacksquare}}

\providecommand{\proofclaimname}{Proof of claim}

\DefineNamedColor{named}{mediumgrey}{RGB}{170,170,170}
\DefineNamedColor{named}{lightgrey}{RGB}{210,210,210}

\newcommand\makebreak{\vspace{0.5\baselineskip}}

\RequirePackage{xspace}
\newcommand{\etal}{et~al.\xspace}


%% file: setup/macros.tex
\newcommand{\defeq}{:=}

\newcommand{\defiff}{:\Leftrightarrow}
\newcommand{\isom}{\cong}

\newcommand{\ra}{\rightarrow}

\newcommand{\card}[1]{|\, #1 \,|}
\newcommand{\sep}{\;|\;}
\newcommand{\divides}{\;|\;}
\newcommand{\fvec}[1]{\mathbf{#1}}	

\DeclareMathOperator{\ord}{ord}	

\newcommand{\eqclass}[2]{[#1]_{#2}}

\makeatletter
\def\imod#1{\allowbreak\mkern10mu({\operator@font mod}\,\,#1)}
\makeatother

\newcommand{\qsep}{\,.\,}

\newcommand{\logic}[1]{\ensuremath{\mathrm{#1}}\xspace}

\newcommand{\formula}[1]{\ensuremath{\mathsf{#1}\xspace}}

\newcommand{\FO}{\logic{FO}}
\newcommand{\FOC}{\logic{FOC}}
\newcommand{\forkp}{\logic{FOR}_p}
\newcommand{\foslvp}{\logic{FOS}_p}
\newcommand{\foslv}[1]{\logic{FOS_{#1}}}

\newcommand{\FPC}{\logic{FPC}}
\newcommand{\FPrk}{\logic{FPR}}
\newcommand{\FP}{\logic{FP}}

\newcommand{\fo}{\logic{FO}\xspace}
\newcommand{\FOdtc}{\logic{DTC}}

\newcommand{\fpc}{\FPC}

\newcommand{\fp}{\FP}

\newcommand{\struct}[1]{{\mathbf #1}}
\newcommand{\structClass}[1]{\mathcal #1}
\newcommand{\univ}[1]{\ensuremath{D  ( \struct{#1} )}} 
\newcommand{\lindstrom}[1]{Q_{\structClass #1}}
\newcommand{\uniformLindstrom }[1]{\langle Q_{\structClass #1} \rangle }

\newcommand{\vocRing}{\tau_{\text{ring}}}
\newcommand{\vocGroup}{\tau_{\text{group}}}
\newcommand{\vocMatrix}{\tau_{\text{mat}}}

\newcommand{\vocMatrixPair}{\tau_{\text{mpair}}}
\newcommand{\vocMatrixOrd}{\tau_{\text{mat}}^\leqslant}
\newcommand{\vocMatrixPairOrd}{\tau_{\text{mpair}}^\leqslant}

\newcommand{\vocLinEqRing}{\tau_{\text{les-r}}}
\newcommand{\vocLinEqRingFixed}[1]{\tau_{\text{les}}(#1)}
\newcommand{\vocLinEqRingOrd}{\tau_{\text{les-r}}^\leqslant}
\newcommand{\vocLinEqGroup}{\tau_{\text{les-g}}}

\newcommand{\logicalReduction}[1]{\ensuremath{\leq_\text{#1}\xspace}}
\newcommand{\logicalTurReduction}[1]{\ensuremath{\leq_\text{#1-\logic T}
\xspace } }

\newcommand{\foleqTur}{\logicalTurReduction{FO}}

\newcommand{\fpleq}{\logicalReduction{\FP}}
\newcommand{\dtcleq}{\logicalReduction{\FOdtc}}

\newcommand{\fpleqTur}{\logicalTurReduction{\FP}}

\newcommand{\problem}[1]{\textsc{#1}\xspace}

\newcommand{\Solve}[1]{\problem{Slv#1}\xspace}

\newcommand{\SolveField}{\Solve{F}}
\newcommand{\SolveGeneralRing}{\Solve{R}}
\newcommand{\SolveRing}{\Solve{CR}}

\newcommand{\SolveLocalRing}{\Solve{LR}}
\newcommand{\SolveOrdRing}{\SolveRing_\leqslant}

\newcommand{\SolvekLocalRing}{\Solve{LR}_k}

\newcommand{\SolveAbGroup}{\Solve{AG}}
\newcommand{\LCON}{\Solve{CycG}}

\newcommand{\GM}{\problem{GM}}

\newcommand{\cclass}[1]{\ensuremath{\mathrm{#1}}}
\newcommand{\Ptime}{\cclass{PTIME}\xspace}
\newcommand{\NPtime}{\cclass{NP}\xspace}
\newcommand{\Logspace}{\cclass{LOGSPACE}\xspace}

\newcommand{\compnc}{\cclass{NC}\xspace}

\newcommand{\N}{\mathbb{N}}
\newcommand{\Z}{\mathbb{Z}}
\newcommand{\Q}{\mathbb{Q}}

\newcommand{\bx}{{\vec{x}}}
\providecommand{\by}{{\vec{y}}}
\renewcommand{\by}{{\vec{y}}}
\newcommand{\ba}{{\vec{a}}}
\newcommand{\bb}{{\vec{b}}}
\newcommand{\bc}{{\vec{c}}}

\newcommand{\tup}[1]{{\vec #1}}

\newcommand{\genl}[2]{\ensuremath{\mbox{\rm\sl GL}_{#1}(#2)}}

\newcommand{\ZeroMatrix}{\ensuremath{\mathbf{0}}}
\newcommand{\IdMatrix}{\ensuremath{\mathbf{1}}}

\newcommand{\inv}[1]{#1^{-1}}

\newcommand{\GF}[1]{\ensuremath{\mathsf{GF}(#1)}}

\newcommand{\GR}[2]{\ensuremath{\mathsf{GR}(#1, #2)}}

\newcommand{\Zm}[1]{{{\mathbb Z}_{#1}}}

\newcommand{\ringBase}[1]{\mathcal{B}(#1)}

\newcommand{\units}[1]{#1^\times}

\newcommand{\function}[1]{\ensuremath{\mathrm{#1}\xspace}}

\makeatletter
\def\Ddots{\mathinner{\mkern1mu\raise\p@
\vbox{\kern7\p@\hbox{.}}\mkern2mu
\raise4\p@\hbox{.}\mkern2mu\raise7\p@\hbox{.}\mkern1mu}}
\makeatother

\newcommand{\lleq}{\leq_{\text{lex}}}

%% file: figures/logical-reductions-solvability.tex
\begin{tikzpicture}[description/.style={fill=white,inner sep=2pt}]

\matrix (m) [matrix of math nodes, row sep=0.8cm,
column sep=0.8cm, text height=2ex, text depth=1ex]
{
	\LCON			&	\SolveAbGroup 	& 	\\
	\SolveOrdRing	&	\SolveRing		&	\SolveGeneralRing \\
	\SolvekLocalRing	&	\SolveField	
& \SolveLocalRing	\\
};

\path[->,font=\tiny]
	(m-2-1) edge node[left] {$\FP$} (m-1-1)
	(m-1-2) edge node[left] {\FOdtc} (m-2-2)
	(m-2-3) edge node[above] {\FOdtc} (m-2-2)
	(m-2-2) edge node[below left=-0.05] {$\FO$-\logic T} (m-3-3)
	(m-3-1) edge node[left] {$\FP$-\logic T} (m-2-1)
;

\path[right hook->]
	(m-1-1) edge node[auto] {} (m-1-2)	
	(m-2-1) edge node[auto] {} (m-2-2)
;
\path[left hook->]
	(m-3-2) edge node[auto] {} (m-3-1)
;

\draw[transform canvas={xshift=3ex},left hook->] (m-3-3) -- (m-2-2);

\end{tikzpicture}

%% file: figures/solvability-nesting.tex
\tikzstyle{blank}=[
	draw=none,
	fill=none
]

\begin{tikzpicture}[x=0.2cm,y=0.2cm,>=stealth]


	\node at (4.7cm,2.3cm) (blank) {\color{white}{x}};

	\node at (-1.8cm,-0.7cm) (blank) {\color{white}{x}};


	\node[blank] (Cell) at (0.5cm, 0.5cm) {};
	\node[blank] (CellNW) at ($(Cell) + (-1mm,1mm)$) {};
	\node[blank] (CellSW) at ($(Cell) + (-1mm,-1mm)$) {};
	\node[blank] (CellSE) at ($(Cell) + (1mm,-1mm)$) {};

	\draw[step=2mm, draw=mediumgrey] (0,0) grid (9,9);

	\draw[draw=black] (0,0) rectangle (1.8cm,1.8cm);

	\draw (-0.7cm,1.2cm) node[font=\footnotesize] {$A^k$};	
	\draw (1.2cm,-0.7cm) node[font=\footnotesize] {$A^\ell$};	

	\draw[draw=black,->,densely dotted,semithick] (-0.4cm,0.5cm) -- ($(Cell) + (-0.5,0)$);	
	\node[font=\footnotesize] at (-0.6cm,0.5cm) {$\tup a$};

	\draw[draw=black,->,densely dotted,semithick] (0.5cm,-0.4cm) -- ($(Cell) + (0,-0.5)$);	
	\node[font=\footnotesize] at (0.5cm,-0.6cm) {$\tup b$};
	
	\draw[fill=lightgrey] (CellNW) rectangle (CellSE) {};
	
	\draw (2.5cm,1cm) node[font=\normalsize] {$\cdot \fvec x = \fvec 1$};

    \draw[decorate,decoration={brace, amplitude=4pt}] 
		(3.2cm,1.8cm) -- (3.2cm,0cm) 
		node[midway, right=2mm,font=\footnotesize]{$\mathcal O(\struct A)$}
		;


	\node[blank] (ISNW) at (-2.9cm,1cm) {};
	\node[blank] (ISTC) at ($(ISNW) + (0.9cm,0)$) {};
	\node[blank] (ISNE) at ($(ISNW) + (1.8cm,0)$) {};
	\node[blank] (ISSE) at ($(ISNE) + (0cm,-1cm)$) {};
	\node[blank] (ISSW) at ($(ISNW) + (0cm,-1cm)$) {};
	
	\draw[densely dotted,draw=mediumgrey,semithick] (ISNW) rectangle (ISSE);
	\draw[densely dotted,draw=mediumgrey,semithick] (ISNE) -- (CellNW);
	\draw[densely dotted,draw=mediumgrey,semithick] (ISSE) -- (CellSW);

    \draw[decorate,decoration={brace, amplitude=4pt}] 
		($(ISSE) + (0,-2mm)$) -- ($(ISSW) + (0,-2mm)$) 
		node[midway, below=2mm,font=\footnotesize]{$\mathcal I(\tup a, \tup b)(\struct A)$}
		;
		
	\node[font=\footnotesize] at ($(ISTC) + (0,-0.5cm)$) {$C_{\tup a \tup b} \cdot \fvec y = \fvec 1$};

	


\end{tikzpicture}